\documentclass[11pt,letterpaper,reqno]{amsart}

\title[]{Shape Dynamics of $N$ Point Vortices on the Sphere}
\author{Tomoki Ohsawa}
\address{Department of Mathematical Sciences, The University of Texas at Dallas, 800 W Campbell Rd, Richardson, TX 75080-3021}
\email{tomoki@utdallas.edu}
\date{\today}

\keywords{Point vortices, Hamiltonian dynamics, Symplectic reduction, Lie--Poisson equation}
\subjclass[2020]{37J37,53D20,70G65,76B47}

\usepackage{graphicx,amssymb,paralist,subfigure,nicematrix}
\usepackage[margin=1in, marginpar=.5in]{geometry}

\usepackage{tikz-cd}

\usepackage[numbers,sort&compress]{natbib}
\usepackage[colorlinks=false]{hyperref}

\usepackage[nameinlink]{cleveref}

\theoremstyle{plain}
\newtheorem{theorem}{Theorem}[section]

\newtheorem{lemma}[theorem]{Lemma}
\newtheorem{proposition}[theorem]{Proposition}

\theoremstyle{definition}

\newtheorem{example}[theorem]{Example}

\theoremstyle{remark}
\newtheorem{remark}[theorem]{Remark}

\def\od#1#2{\frac{d#1}{d#2}}
\def\pd#1#2{\frac{\partial #1}{\partial #2}}
\def\fd#1#2{\frac{\delta #1}{\delta #2}}

\def\parentheses#1{\!\left(#1\right)}
\def\brackets#1{\!\left[#1\right]}
\def\braces#1{\!\left\{#1\right\}}

\def\tr{\mathop{\mathrm{tr}}\nolimits}
\def\diag{\operatorname{diag}}

\def\norm#1{{\left\|#1\right\|}}
\def\abs#1{\left|#1\right|}
\def\DS{\displaystyle}
\def\R{\mathbb{R}}
\def\C{\mathbb{C}}

\def\N{\mathbb{N}}
\def\defeq{\mathrel{\mathop:}=}
\def\eqdef{=\mathrel{\mathop:}}
\def\setdef#1#2{\left\{ #1 \ |\ #2 \right\}}
\def\ip#1#2{{\left\langle#1,#2\right\rangle}}

\renewcommand{\Re}{\operatorname{Re}}
\renewcommand{\Im}{\operatorname{Im}}

\def\rmi{{\rm i}}

\def\d{\mathbf{d}}
\def\ins#1{{\bf i}_{#1}}
\def\PB#1#2{\left\{#1,#2\right\}}
\newcommand\Ad{\operatorname{Ad}}
\newcommand\ad{\operatorname{ad}}

\def\SO{\mathsf{SO}}

\def\U{\mathsf{U}}
\def\SU{\mathsf{SU}}

\def\u{\mathfrak{u}}
\def\so{\mathfrak{so}}
\def\su{\mathfrak{su}}

\begin{document}

\footskip=.6in

\begin{abstract}
  We give a geometric account of the relative motion or the shape dynamics of $N$ point vortices on the sphere exploiting the $\mathsf{SO}(3)$-symmetry of the system. The main idea is to bypass the technical difficulty of the $\mathsf{SO}(3)$-reduction by first lifting the dynamics from $\mathbb{S}^{2}$ to $\mathbb{C}^{2}$. We then perform the $\mathsf{U}(2)$-reduction using a dual pair to obtain a Lie--Poisson dynamics for the shape dynamics. This Lie--Poisson structure helps us find a family of Casimirs for the shape dynamics. We further reduce the system by $\mathbb{T}^{N-1}$-symmetry to obtain a Poisson structure for the shape dynamics involving fewer shape variables than those of the previous work by Borisov and Pavlov. As an application of the shape dynamics, we prove that the tetrahedron relative equilibria are stable when all of their circulations have the same sign, generalizing some existing results on tetrahedron relative equilibria of identical vortices.
\end{abstract}

\maketitle

\section{Introduction}
\subsection{Dynamics of Point Vortices on Sphere}
Consider $N$ point vortices on the two-sphere $\mathbb{S}^{2}_{R} \subset \R^{3}$ with (fixed) radius $R > 0$ centered at the origin.
Let $\{ \mathbf{x}_{i} \in \mathbb{S}^{2}_{R} \}_{i=1}^{N}$ be the positions of the point vortices with circulations $\{ \Gamma_{i} \}_{i=1}^{N}$.
Then the equations of motion of the point vortices are
\begin{equation}
  \label{eq:vortices_on_sphere}
  \dot{\mathbf{x}}_{i}
  = \frac{1}{2\pi R} \sum_{\substack{1 \le j \le N\\j \neq i}} \Gamma_{j} \frac{ \mathbf{x}_{j} \times \mathbf{x}_{i} }{ |\mathbf{x}_{i} - \mathbf{x}_{j}|^{2} }
\end{equation}
for $i \in \{1, \dots, N\}$; see, e.g., \citet{Bo1977}, \citet{KiOk1987}, and \citet[Chapter~4]{Ne2001}.

This system of equations is Hamiltonian in the following sense:
Let $\Omega_{i}$ be the area form of the $i$-th copy of $\mathbb{S}^{2}_{R}$ and define the following two-form on $(\mathbb{S}^{2}_{R})^{N}$:
\begin{equation}
  \label{eq:Omega-S^2}
  \Omega_{\mathbb{S}^{2}_{R}} \defeq \sum_{i=1}^{N} \Gamma_{i} \pi_{i}^{*}\Omega_{i}
  \quad
  \text{with}
  \quad
  \Omega_{i}(\mathbf{x}_{i})(\mathbf{v}_{i}, \mathbf{w}_{i}) \defeq \frac{1}{R} \mathbf{x}_{i} \cdot (\mathbf{v}_{i} \times \mathbf{w}_{i})
\end{equation}
where $\pi_{i}\colon (\mathbb{S}^{2}_{R})^{N} \to \mathbb{S}^{2}_{R}$ is the projection to the $i$-th copy.
The corresponding Poisson bracket on $(\R^{3})^{N}$ (see \Cref{sec:recovering} for details) is, for all smooth $F, H\colon (\R^{3})^{N} \to \R$,
\begin{equation}
  \label{eq:PB-R3}
  \PB{F}{H}_{\R^{3}}(\mathbf{x}) \defeq \sum_{i=1}^{N}\frac{R}{\Gamma_{i}} \mathbf{x}_{i} \cdot \parentheses{ \pd{F}{\mathbf{x}_{i}} \times \pd{H}{\mathbf{x}_{i}} }.
\end{equation}
Define the Hamiltonians on $(\mathbb{S}^{2}_{R})^{N}$ and $(\R^{3})^{N}$ as follows:
\begin{subequations}
  \label{eq:H}
  \begin{equation}
    \label{eq:H-S^2}
    H_{\mathbb{S}^{2}_{R}}(\mathbf{x}_{1}, \dots, \mathbf{x}_{N})
    \defeq -\frac{1}{4\pi R^{2}} \sum_{1\le i < j \le N} \Gamma_{i} \Gamma_{j} \ln \parentheses{
      2 (R^{2} - \mathbf{x}_{i} \cdot \mathbf{x}_{j} )
    },
  \end{equation}
  and
  \begin{equation}
    \label{eq:H-R^3}
    H_{\R^{3}}(\mathbf{x}_{1}, \dots, \mathbf{x}_{N}) \defeq -\frac{1}{4\pi R^{2}} \sum_{1\le i < j \le N} \Gamma_{i} \Gamma_{j} \ln \parentheses{
      |\mathbf{x}_{i} - \mathbf{x}_{j}|^{2}
    }.
  \end{equation}
\end{subequations}
Note that the former is the restriction of the latter to $(\mathbb{S}^{2}_{R})^{N}$.
Then we obtain \eqref{eq:vortices_on_sphere} as a Hamiltonian system on $(\mathbb{S}^{2}_{R})^{N}$ or $(\R^{3})^{N}$ as follows:
\begin{equation}
  \label{eq:vortices_on_sphere-Hamiltonian}
  \ins{X}{ \Omega_{\mathbb{S}^{2}_{R}} } = \d H_{\mathbb{S}^{2}_{R}}
  \quad
  \text{or}
  \quad
  \dot{\mathbf{x}}_{i} = \PB{ \mathbf{x}_{i} }{ H_{\R^{3}} },
\end{equation}
where $X$ is a vector field on $(\mathbb{S}^{2}_{R})^{N}$.

\begin{remark}
  The Hamiltonians~\eqref{eq:H} have singularities at the collision points, i.e., $\mathbf{x}_{i} = \mathbf{x}_{j}$ with $i \neq j$.
  Following \citet[Remark~1.1]{Ki1988}, we will ignore this issue for now because the concrete expression for the Hamiltonian does not affect the geometry of our problem as long as it possesses the $\SO(3)$-symmetry described below.
  Once we obtain the Hamiltonian for the shape dynamics, we may remove the singularities by imposing conditions on the corresponding variables accordingly.
  Alternatively, one may also introduce a regularization parameter to remove the singularities from the outset; see, e.g., \citet{VaLe2014}.
\end{remark}

The dynamics of point vortices on the sphere has been studied quite extensively because it is not only interesting mathematically but also has geophysical and astrophysical applications.
For example, \citet{DiPo1998} and \citet{Ki1999} studied the motion of a vortex pair ($N = 2$), and \citet{KiNe1998} solved the equations of relative motion (see \eqref{eq:relative_motion} below) for $N = 3$.
\citet{BoLe1998} and \citet{Sa1999} studied the integrable three-vortex motions on the sphere; see also \citet{Sa2007} for an integrable four-vortex motion on sphere with zero moment of vorticity, and \citet{SaYa2008a,SaYa2008b} for studies on chaotic motions of $N$ point vortices on the sphere.

One can also generalize the basic equations~\eqref{eq:vortices_on_sphere} to those vortices on a rotating sphere, and their dynamics has been studied in, e.g., \citet{NeSh2006,JaNe2006,NeSa2007} and \citet{La2005}.

Stability of fixed and relative equilibria of point vortices on the sphere is one of the major topics of research as well.
The linear stability of rings of identical vortices was studied by \citet{PoDr1993}, and its nonlinear stability by \citet{BoCa2003} (see also \citet{LaMoRo2011}).
\citet{LiMoRo2001} proved the existence of a number of relative equilibria of identical vortices, and \citet{LaMoRo2011} studied nonlinear stability of many different types of relative equilibria involving one or two ring of vortices---each consisting of identical vortices---with and without one or two polar vortices; see also \citet{La2001} and \citet{BoSi2019}.
See also \cite{GaGa2022} for the existence of periodic orbits of $N$ identical vortices and small nonlinear oscillations near the Platonic solid equilibria, and also \citet{MoTo2013} for the bifurcation of the heptagon equilibrium with the Gaussian curvature being the parameter.

\subsection{Relative Motion and Shape Dynamics}
The focus of this paper is the relative motion or the \textit{shape dynamics} of the point vortices, i.e., we are interested in the set of equations that governs the evolution of the ``shape'' or relative positions of the point vortices---regardless of where the vortices are located on the sphere.
For example, for $N = 3$, it is the dynamics of the shape of the triangle formed by the three point vortices, regardless of its position and orientation on the sphere.

Defining the inter-vortex (Euclidean) distance
\begin{equation*}
  \ell_{ij} \defeq | \mathbf{x}_{i} - \mathbf{x}_{j} |
\end{equation*}
for $i, j \in \{1, \dots, N\}$ with $i \neq j$ and the (signed) volume
\begin{equation*}
  V_{ijk} \defeq \mathbf{x}_{i} \cdot (\mathbf{x}_{j} \times \mathbf{x}_{k})
\end{equation*}
of the parallelepiped formed by vectors $\mathbf{x}_{i}, \mathbf{x}_{j}, \mathbf{x}_{k}$ for $i,j,k \in \{1, \dots, N\}$ with $i \neq j \neq k$, we can derive the \textit{equations of relative motion}
\begin{equation}
  \label{eq:relative_motion}
  \od{}{t} \ell_{ij}^{2} = \frac{1}{\pi R} \sum_{\substack{1 \le k \le N\\k \neq i \neq j}}^{N} \Gamma_{k} V_{ijk} \left( \frac{1}{\ell_{jk}^{2}} - \frac{1}{\ell_{ki}^{2}} \right)
\end{equation}
from \eqref{eq:vortices_on_sphere}; see, e.g., \citet[Section~4.2]{Ne2001}.

\subsection{Hamiltonian Formulation of Shape Dynamics}
Given that the original equation~\eqref{eq:vortices_on_sphere} is a Hamiltonian system, a natural question to ask is whether the equations~\eqref{eq:relative_motion} of relative motion or shape dynamics are also a Hamiltonian system.
In fact, \citet{BoPa1998} derived the Poisson bracket for the above ``internal'' variables $\{ \ell_{ij} \}_{1\le i < j \le N} \cup \{ V_{ijk} \}_{1\le i < j < k \le N}$ in a direct manner from the Poisson bracket~\eqref{eq:PB-R3} for the original dynamics~\eqref{eq:vortices_on_sphere}.  

A more geometric perspective of this question is the following:
Intuitively, it is clear that the dynamics of $N$ point vortices governed by \eqref{eq:vortices_on_sphere} would have $\SO(3)$-symmetry under the rotational action
\begin{equation*}
  \SO(3) \times (\mathbb{S}^{2}_{R})^{N} \to (\mathbb{S}^{2}_{R})^{N};
  \qquad
  (A, (\mathbf{x}_{1}, \dots, \mathbf{x}_{N})) \mapsto ( A \mathbf{x}_{1}, \dots, A \mathbf{x}_{N}).
\end{equation*}
This action is clearly symplectic with respect to the symplectic form~\eqref{eq:Omega-S^2} because the volume form of each sphere is invariant under the rotational action.
One also sees that the Hamiltonian~\eqref{eq:H-S^2} is $\SO(3)$-invariant as well.
Taking the quotient by $\SO(3)$, we identify all the configurations of the vortices that are congruent to each other as a single ``shape''.
So if we could perform the symplectic reduction (see \citet{MaWe1974} and \cite[Sections~1.1 and 1.2]{MaMiOrPeRa2007}) of \eqref{eq:vortices_on_sphere} by the $\SO(3)$-symmetry, then the resulting reduced dynamics would be essentially the equations~\eqref{eq:relative_motion} of relative motion.
Such a geometric picture of shape dynamics has been also applied to the $N$-body problem of classical mechanics (see, e.g., \citet{Iw1987}, \citet{Mo2015}, and references therein) and also point vortices on the plane; see, e.g., \citet{KoPiRoGo1985}, \cite{BoBoMa1999}, and \cite{Oh2019d}.

Unfortunately, the reduction by $\SO(3)$-symmetry is quite intricate.
The momentum map associated with the above $\SO(3)$-action gives the following well-known invariant:
\begin{equation*}
  \mathbf{I}\colon (\mathbb{S}^{2}_{R})^{N} \to \so(3)^{*} \cong \R^{3};
  \qquad
  \mathbf{I}(\mathbf{x}_{1}, \dots, \mathbf{x}_{N}) \defeq \frac{1}{R}\sum_{i=1}^{N} \Gamma_{i} \mathbf{x}_{i}.
\end{equation*}
The difficulty is that the reduced space or the Marsden--Weinstein quotient $\mathbf{I}^{-1}(\mathbf{c})/\SO(3)_{\mathbf{c}}$ with $\mathbf{c} \in \R^{3}$ is tricky to work with when describing the reduced dynamics, where $\SO(3)_{\mathbf{c}}$ stands for the isotropy group $\setdef{ A \in \SO(3) }{ A \mathbf{c} = \mathbf{c} }$.
While \citet{Ki1988} found some topological invariants of the reduced space, the focus was rather on the topology of the space than the dynamics.
Indeed, it is difficult to find coordinates for the reduced space in general, and concrete treatments of the reduced dynamics are limited to some special cases; see \citet{PeMa1998} for the corresponding Poisson reduction in the special case with $N = 3$ and \citet{Li1998} for an explicit treatment of the reduction for $N = 4$.

\subsection{Main Results and Outline}
Our main contribution is the geometric treatment of the shape dynamics exploiting the $\SO(3)$-symmetry mentioned above.
Specifically, we proceed as follows to sidestep the difficulty of the $\SO(3)$-reduction; see also \Cref{fig:lift_and_reduce}.
\begin{enumerate}[\sf 1.]
\item We first lift the dynamics of vortices from $\mathbb{S}^{2}_{R}$ to $\C^{2}$ in \Cref{sec:lift}.
  For $N$ vortices, the lifted dynamics is then in $(\C^{2})^{N}$, which is identified with the space $\C^{2 \times N}$ of $2 \times N$ complex matrices.
  The lifted dynamics possesses a $\mathbb{T}^{N} \defeq \mathbb{S}^{1} \times \dots \times \mathbb{S}^{1}$ ($N$ copies)-symmetry, and the symplectic reduction by the symmetry recovers the vortex dynamics on the sphere; see \Cref{prop:reduction_of_lifted_dynamics}.
  \smallskip
\item In \Cref{sec:U(2)-reduction}, we perform a $\U(2)$-reduction of the lifted dynamics using a dual pair of \citet{SkVi2019} defined on $\C^{2 \times N}$.
  This essentially corresponds to the $\SO(3)$-reduction of the original dynamics because its $\SU(2)$ subgroup symmetry gives the $\SO(3)$-symmetry of the original dynamics on the sphere.
  The use of the dual pair facilitates the reduction, because the dual pair essentially allows one to embed the reduced space to the dual of a Lie algebra, yielding a Lie--Poisson equation for the reduced dynamics; see, e.g., \citet{We1983}, \citet[Section~IV.7]{LiMa1987}, and \citet[Chapter~11]{OrRa2004}.
  In other words, instead of having the reduced dynamics in a complicated quotient manifold, the reduced dynamics is given by an ODE on a vector space.
  It also helps us find a family of Casimirs associated with the Lie--Poisson structure; see \Cref{prop:Casimirs}.
  \smallskip
\item In \Cref{sec:further_reduction}, we further reduce the Lie--Poisson dynamics using the $\mathbb{T}^{N-1}$-symmetry to get rid of the extra symmetry picked up by the lifting\footnote{The reason why we have $\mathbb{T}^{N-1}$-symmetry as opposed to $\mathbb{T}^{N}$ is that one copy of $\mathbb{S}^{1}$ is taken care of in the $\U(2)$-reduction in the previous step.}.
  The resulting Poisson structure gives a Hamiltonian formulation of the shape dynamics; see \Cref{thm:main}.
\end{enumerate}

\begin{figure}[htbp]
  \centering
  \begin{tikzcd}[ampersand replacement=\&, column sep=5ex, row sep=9ex]
    \& (\C^{2})^{N} \cong \C^{2\times N} \arrow[swap, shift right=.5ex]{ld}{ \substack{\text{Reduction by $\mathbb{T}^{N}$}\\\text{(\Cref{ssec:S^1-reduction})}} \hspace{-1ex} } \arrow[red!55!black]{rd}{ \hspace{-2ex}\substack{ \text{Reduction by $\U(2)$} \\ \text{(\Cref{sec:U(2)-reduction})} }} \& \\
    (\mathbb{S}^{2}_{R})^{N} \arrow[swap, dashed]{rd}{\text{Reduction by $\SO(3)$}} \arrow[swap, red!55!black, shift right=0.5ex]{ur}{ \hspace{-3ex}\substack{ \text{Lifting}\\\text{(\Cref{ssec:lifting})} } } \& \& \mathcal{O} \subset \u(D_{\Gamma})^{*} \arrow[red!55!black]{ld}{ \substack{ \text{Reduction by $\mathbb{T}^{N-1}$} \\ \text{(\Cref{sec:further_reduction})} } } \\
    \& \text{shape space} \& 
  \end{tikzcd}
  \caption{Instead of reducing the dynamics on $(\mathbb{S}^{2}_{R})^{N}$ by $\SO(3)$ directly, first lift it to $(\C^{2})^{N}$ (which picks up $\mathbb{T}^{N}$-symmetry) and then apply reduction by $\U(2)$ (which is facilitated by a dual pair); this results in a Lie--Poisson dynamics in a coadjoint orbit $\mathcal{O} \subset \u(D_{\Gamma})^{*}$.
    We may then further reduce the system by $\mathbb{T}^{N-1}$-symmetry to get rid of the extra symmetry picked up in the lifting process.
  }
  \label{fig:lift_and_reduce}
\end{figure}
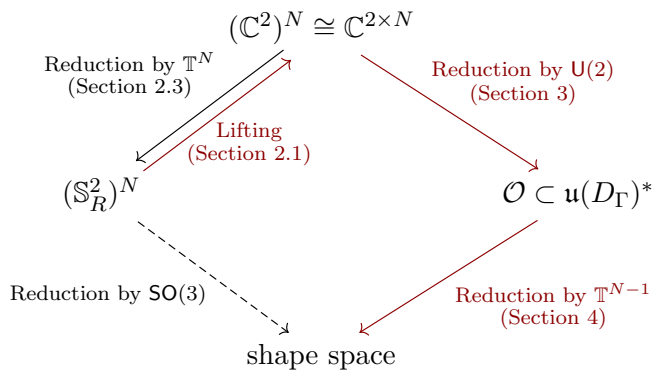

This geometric treatment results in fewer variables for the shape dynamics compared to those ``internal'' variables of \citet{BoPa1998}.
In fact, our shape dynamics is described using $(N - 1)^{2}$ variables, but it turns out that $(N - 1)(N - 2)/2$ of those implicitly depend on the rest; hence $N(N - 1)/2$ variables essentially.
On the other hand, the number of the ``internal'' variables of \cite{BoPa1998} is $N(N^{2} - 1)/6$.

Another advantage of our formulation is that we can find a family of Casimirs exploiting the underlying algebraic structure of the Lie--Poisson bracket on $\mathcal{O}$.
This is not easy with the Poisson bracket of \cite{BoPa1998} because they obtained it ``by hand'', i.e., the algebraic structure of their Poisson bracket is not clear.

Note that our parametrization does not in general give actual coordinate charts for the shape space.
Indeed, since the original dynamics is $2N$-dimensional, the shape space after the $\SO(3)$-reduction would \textit{not} have dimensions such as $(N - 1)^{2}$.
Instead, we sidestep the difficulty of directly dealing with the shape space by describing the shape dynamics using the $(N - 1)^{2}$ coordinates for the ambient space for the shape space.
The reason for the increase in the number of variables is that the $\U(2)$-reduced dynamics in the coadjoint orbit $\mathcal{O}$ (see \Cref{fig:lift_and_reduce}) is described in terms of the coordinates for $\u(D_{\Gamma})^{*}$; note that this is generally the case with Lie--Poisson dynamics.

Although this redundancy of shape variables is certainly a drawback, the resulting shape dynamics provides a means to analyze the stability of relative equilibria (i.e., the stability of the shape formed by the vortices).
To demonstrate this idea, we apply the energy--Casimir method to our shape dynamics with $N = 4$ and find a sufficient condition for the stability of tetrahedron relative equilibria in \Cref{sec:application}.
Our result concerns the non-identical case, i.e., $\Gamma_{1} \neq \Gamma_{2} \neq \Gamma_{3} \neq \Gamma_{4}$, and generalizes those results of \citet{Ku2004} and \citet{MeNeOs2010} for the identical case with $\Gamma_{1} = \Gamma_{2} = \Gamma_{3} = \Gamma_{4}$.
We also mention in passing that \citet{PeMa1998} used the energy--momentum method to find a sufficient condition for stability of non-identical equilateral triangle relative equilibria, i.e., $N = 3$ with $\Gamma_{1} \neq \Gamma_{2} \neq \Gamma_{3}$.
So our result is also an extension of theirs to $N = 4$ as well.

\section{Lifted Vortex Dynamics in $\C^{2}$}
\label{sec:lift}
We would like to first lift the vortex dynamics from $\mathbb{S}^{2}_{R}$ to $\C^{2}$.
This idea is inspired by \citet{VaLe2014}, where they lift the dynamics from $\mathbb{S}^{2}_{R}$ to $\mathbb{S}^{3}_{\sqrt{R}}$ via the Hopf fibration $\mathbb{S}^{3}_{\sqrt{R}} \to \mathbb{S}^{2}_{R}$.
We shall show that our approach naturally gives rise to the Hopf fibration by identifying the reduced space $\mathbb{S}^{2}_{R} = \mathbb{S}^{3}_{\sqrt{R}}/\mathbb{S}^{1}$ as a Marsden--Weinstein quotient.

\subsection{Vortex Equations in $\C^{2}$}
\label{ssec:lifting}
Let us show how the lifting from \Cref{fig:lift_and_reduce} works.
Since the Hopf fibration map gives rise to a map from $\C^{2}$ to $\R^{3}$, one may relate the distance in $\C^{2}$ with that in $\R^{3}$ as shown in \eqref{eq:C^2_distance-R^3} of \Cref{ssec:ip-C^2_R^3}.
Replacing the inter-vortex distance in $\R^{3}$ in the Hamiltonian~\eqref{eq:H-R^3} by the corresponding distance in $\C^{2}$ using \eqref{eq:C^2_distance-R^3}, we define a Hamiltonian $H\colon (\C^{2})^{N} \to \R$ as
\begin{equation}
  \label{eq:H-C^2}
  H(\boldsymbol{\varphi})
  \defeq
  -\frac{1}{4\pi R^{2}} \sum_{1\le i < j \le N} \Gamma_{i} \Gamma_{j} \ln \brackets{
    \left( \norm{\varphi_{i}}^{2} + \norm{\varphi_{j}}^{2} \right)^{2} - 4|\varphi_{i}^{*}\varphi_{j}|^{2}
  },
\end{equation}
where we used the shorthand
\begin{equation*}
  \boldsymbol{\varphi} = (\varphi_{1}, \dots , \varphi_{N}) \in (\C^{2})^{N},
\end{equation*}
and defined the norm $\norm{\varphi} \defeq \sqrt{\varphi^{*}\varphi}$ induced by the natural inner product on $\C^{2}$.
We also write
\begin{equation*}
  \varphi_{i} =
  \begin{bmatrix}
    z_{i} \\
    u_{i}
  \end{bmatrix}
  \quad
  \text{with}
  \quad
  z_{i}, u_{i} \in \C
  \quad
  \forall i \in \{1, \dots, N\}.
\end{equation*}

We define a symplectic form $\Omega$ on $(\C^{2})^{N}$ as follows:
\begin{equation*}
  \Omega \defeq -\frac{2}{R}\sum_{i=1}^{N} \Gamma_{i} \Im\parentheses{ \d\varphi_{i}^{*} \wedge \d\varphi_{i} },
\end{equation*}
or $\Omega = -\d\Theta$ with
\begin{equation}
  \label{eq:Theta-C2}
  \Theta \defeq -\frac{2}{R}\sum_{i=1}^{N} \Gamma_{i} \Im\parentheses{ \varphi_{i}^{*} \d\varphi_{i} }.
\end{equation}
Then the Hamiltonian vector field (``c.c.'' stands for the complex conjugate of the preceding term)
\begin{equation*}
  X = \dot{\varphi_{i}} \pd{}{\varphi_{i}} + \text{c.c.}
\end{equation*}
defined by the Hamiltonian system $\ins{X}{\Omega} = \d H$ gives the following Schr\"odinger-like lifted vortex equation on $\C^{2}$ for $i = 1, \dots, N$:
\begin{equation}
  \label{eq:lifted_dynamics}
  \Gamma_{i}\dot{\varphi}_{i} = -\frac{\rmi}{2} \pd{H}{\varphi_{i}^{*}}.
\end{equation}

\subsection{$\mathbb{T}^{N}$-symmetry and Momentum Map}
The above lifted vortex equations~\eqref{eq:lifted_dynamics} possesses a symmetry under the natural action of the torus
\begin{equation*}
  \mathbb{T}^{N} = ( \mathbb{S}^{1} )^{N}
  = \setdef{ e^{\rmi\boldsymbol{\theta}} \defeq (e^{\rmi\theta_{1}}, \dots, e^{\rmi\theta_{N}}) }{ \theta_{i} \in [0,2\pi) \text{ for } i = 1, \dots, N }
\end{equation*}
defined as
\begin{equation*}
  \mathbb{T}^{N} \times (\C^{2})^{N} \to (\C^{2})^{N};
  \qquad
  (e^{\rmi\boldsymbol{\theta}}, \boldsymbol{\varphi})
  \mapsto \parentheses{
    e^{\rmi\theta_{1}}\varphi_{1}, \dots, e^{\rmi\theta_{N}}\varphi_{N}
  }
  \eqdef e^{\rmi\boldsymbol{\theta}} \cdot \boldsymbol{\varphi}.
\end{equation*}
Indeed, one easily sees that the one-form~\eqref{eq:Theta-C2} and the Hamiltonian~\eqref{eq:H-C^2} are invariant under the action.
Let $\boldsymbol{\omega} \in (T_{1}\mathbb{S}^{1})^{N} \cong \R^{N}$.
Its corresponding infinitesimal generator is
\begin{equation*}
  \boldsymbol{\omega}(\boldsymbol{\varphi})
  = \left.\od{}{s} e^{\rmi s\boldsymbol{\omega}} \cdot \boldsymbol{\varphi} \right|_{s=0}
  = \rmi(\omega_{1}\varphi_{1}, \dots, \omega_{N}\varphi_{N}).
\end{equation*}
Hence the associated momentum map $\mathbf{J}\colon (\C^{2})^{N} \to (\R^{N})^{*} \cong \R^{N}$ satisfies
\begin{align*}
  \mathbf{J}(\boldsymbol{\varphi}) \cdot \boldsymbol{\omega}
  &= \Theta(\boldsymbol{\varphi}) \cdot \boldsymbol{\omega}(\boldsymbol{\varphi}) \\
  &= -\frac{2}{R} \sum_{i=1}^{N} \Gamma_{i} \Im\parentheses{ \rmi\,\omega_{i} \varphi_{i}^{*}\varphi_{i} } \\
  &= -\frac{2}{R} \sum_{i=1}^{N} \Gamma_{i} \omega_{i} \norm{\varphi_{i}}^{2} \\
  &= -\frac{2}{R} \parentheses{ \Gamma_{1}\norm{\varphi_{1}}^{2}, \dots, \Gamma_{N}\norm{\varphi_{N}}^{2} } \cdot \boldsymbol{\omega}.
\end{align*}
As a result, we obtain
\begin{equation*}
  \mathbf{J}(\boldsymbol{\varphi}) = -\frac{2}{R} \parentheses{ \Gamma_{1}\norm{\varphi_{1}}^{2}, \dots, \Gamma_{N}\norm{\varphi_{N}}^{2} }.
\end{equation*}

\subsection{$\mathbb{S}^{1}$-Reduction}
\label{ssec:S^1-reduction}
Now let us explain the $\mathbb{T}^{N}$-reduction part from \Cref{fig:lift_and_reduce}.
Since the Hamiltonian~\eqref{eq:H-C^2} is invariant under the above $\mathbb{T}^{N}$-action, its associated momentum map $\mathbf{J}$ is an invariant of \eqref{eq:lifted_dynamics}.
Therefore, setting $\boldsymbol{\Gamma} \defeq (\Gamma_{1}, \dots, \Gamma_{N}) \in \R^{N}$, the level set
\begin{equation*}
  \mathbf{J}^{-1}(-2\boldsymbol{\Gamma})
  = \mathbb{S}^{3}_{\sqrt{R}} \times \dots \times \mathbb{S}^{3}_{\sqrt{R}}
  = \parentheses{ \mathbb{S}^{3}_{\sqrt{R}} }^{N}
\end{equation*}
is an invariant manifold of the dynamics.
In fact, one can show the following:

\begin{proposition}
  \label{prop:reduction_of_lifted_dynamics}
  The symplectic reduction of $(\C^{2})^{N}$ by the above $\mathbb{T}^{N}$-symmetry yields the Marsden--Weinstein quotient
  \begin{equation*}
    \mathbf{J}^{-1}(-2\boldsymbol{\Gamma})/\mathbb{T}^{N}
    = \parentheses{ \mathbb{S}^{3}_{\sqrt{R}}/\mathbb{S}^{1} } \times \dots \times \parentheses{ \mathbb{S}^{3}_{\sqrt{R}}/\mathbb{S}^{1} }
    = \parentheses{ \mathbb{S}^{2}_{R} }^{N}.
  \end{equation*}
  In addition, the lifted dynamics~\eqref{eq:lifted_dynamics} is reduced to the point vortex dynamics~\eqref{eq:vortices_on_sphere} on $\parentheses{ \mathbb{S}^{2}_{R} }^{N}$.
\end{proposition}
\begin{proof}
  See \Cref{sec:recovering}.
\end{proof}

\section{$\U(2)$-Reduction of $N$-vortex Dynamics in $\C^{2}$}
\label{sec:U(2)-reduction}
This section corresponds to the $\U(2)$-reduction part in \Cref{fig:lift_and_reduce}.
The lifted dynamics turns out to possess a $\U(2)$-symmetry, and its $\SU(2)$ subgroup symmetry corresponds to the $\SO(3)$-symmetry of the original system on $\mathbb{S}^{2}_{R}$.
The advantage of the lifted dynamics is that the corresponding $\U(2)$-action on $\C^{2}$ is much more tractable compared to the $\SO(3)$-action on $\mathbb{S}^{2}_{R}$ when it comes to the symplectic reduction.
We exploit the dual pair of \citet{SkVi2019} to show that the $\U(2)$-reduced dynamics is a Lie--Poisson dynamics.

The upshot of this section is that we have a pair of momentum maps on $(\C^{2})^{N} \cong \C^{2\times N}$:
\begin{equation*}
  \begin{tikzcd}
    \u(2)^{*} & \C^{2\times N} \arrow[swap]{l}{\mathbf{K}} \arrow{r}{\mathbf{L}} & \u(N)_{\Gamma}^{*}.
  \end{tikzcd}
\end{equation*}
We shall explain the notation and the details along the way, but the dual pair implies that the reduction by $\U(2)$-symmetry of the lifted dynamics~\eqref{eq:lifted_dynamics} in $\C^{2\times N}$ yields a Lie--Poisson dynamics in $\u(N)_{\Gamma}^{*}$.

\subsection{$\U(2)$-Symmetry of Lifted $N$-vortex Dynamics}
Let us first identify $(\C^{2})^{N}$ with the space of $2 \times N$ complex matrices as follows:
\begin{equation*}
  (\C^{2})^{N} \to \C^{2 \times N};
  \qquad
  \boldsymbol{\varphi} = (\varphi_{1}, \dots , \varphi_{N}) \mapsto \Phi = [\varphi_{1} \dots \varphi_{N}].
\end{equation*}
Then we may rewrite the canonical one-form~\eqref{eq:Theta-C2} as
\begin{equation}
  \label{eq:Theta-C2N}
  \Theta(\Phi) = -\frac{2}{R}\Im\parentheses{ \tr\parentheses{ D_{\Gamma} \Phi^{*}\d\Phi } },
\end{equation}
where we defined
\begin{equation}
  \label{eq:D_Gamma}
  D_{\Gamma} \defeq \diag(\Gamma_{1}, \dots, \Gamma_{N})
  =
  \begin{bNiceMatrix}
    \Gamma_{1} & 0 & \Cdots & 0 \\
    0 & \Gamma_{2} & \Ddots & \Vdots \\
    \Vdots & \Ddots & \Ddots & 0 \\
    0 & \Cdots & 0 & \Gamma_{N}
  \end{bNiceMatrix}.
\end{equation}

Now consider the (left) $\U(2)$-action on $\C^{2 \times N}$ defined as
\begin{equation}
  \label{eq:U(2)-action}
  \U(2) \times \C^{2 \times N} \to \C^{2 \times N};
  \qquad
  (Y, \Phi) \mapsto Y\Phi.
\end{equation}
It is clear that this action leaves $\Theta$ invariant, and hence is a canonical action with respect to the symplectic form $\Omega = -\d\Theta$.
It is also easy to see that the Hamiltonian~\eqref{eq:H-C^2} is invariant under the action as well; hence $\U(2)$ is a symmetry group of the lifted dynamics~\eqref{eq:lifted_dynamics}.

\begin{remark}
  \label{rem:U(2)-SU(2)}
  As is well known, it is not the $\U(2)$-action but the $\SU(2)$-action on $\C^{2}$ that gives rise to the natural $\SO(3)$-action on $\R^{3}$.
  So the above $\U(2)$-symmetry does not exactly correspond to the rotational symmetry of the point vortices dynamics on $\mathbb{S}^{2}_{R}$.
  In fact, the above $\U(2)$-symmetry combines the global phase symmetry (see \Cref{rem:U(2)-T^{N-1}} below) and the rotational symmetry of the system.
  We perform the $\U(2)$-reduction here because the dual pair to be employed below is readily available with this setting, whereas it is unknown with $\SU(2)$.
\end{remark}

\begin{lemma}
  The momentum map $\mathbf{K}\colon \C^{2 \times N} \to \u(2)^{*}$ associated with the above $U(2)$-action~\eqref{eq:U(2)-action} is
  \begin{equation}
    \label{eq:K}
    \mathbf{K}(\Phi) = -\frac{\rmi}{R} \Phi D_{\Gamma} \Phi^{*}
    = -\frac{\rmi}{R} \sum_{i=1}^{N} \Gamma_{i} \varphi_{i} \varphi_{i}^{*}
    = -\frac{\rmi}{R} \sum_{i=1}^{N} \Gamma_{i}
    \begin{bmatrix}
      |z_{i}|^{2} & z_{i}\bar{u}_{i} \\
      \bar{z}_{i}u_{i} & |u_{i}|^{2}
    \end{bmatrix}.
  \end{equation}
\end{lemma}
\begin{proof}
  We equip $\u(2)$ with the inner product
  \begin{equation}
    \label{eq:ip-u2}
    \ip{\xi}{\eta} \defeq 2\tr(\xi^{*} \eta),
  \end{equation}
  and identify the dual $\u(2)^{*}$ with $\u(2)$ via this inner product.
  Since the infinitesimal generator of an arbitrary element $\xi \in \u(2)$ is
  \begin{equation*}
    \xi_{\C^{2 \times N}}(\Phi) = \xi\Phi,
  \end{equation*}
  the associated momentum map $\mathbf{K}\colon \C^{2 \times N} \to \u(2)^{*}$ satisfies
  \begin{align*}
    \ip{ \mathbf{K}(\Phi) }{ \xi }
    &= \Theta(\Phi) \cdot \xi_{\C^{2 \times N}}(\Phi) \\
    &= -\frac{2}{R}\Im\parentheses{ \tr\parentheses{ D_{\Gamma} \Phi^{*} \xi \Phi } } \\
    &= -\frac{2}{R}\Im\parentheses{ \tr\parentheses{ \Phi D_{\Gamma} \Phi^{*} \xi } } \\
    &= \frac{2}{R}\rmi \tr\parentheses{ \Phi D_{\Gamma} \Phi^{*} \xi } \\
    &= \frac{2}{R}\tr\parentheses{ (-\rmi \Phi D_{\Gamma} \Phi^{*})^{*}\xi } \\
    &= \ip{ -\frac{\rmi}{R} \Phi D_{\Gamma} \Phi^{*} }{ \xi }.
  \end{align*}
  Hence the expression~\eqref{eq:K} follows.
\end{proof}

\subsection{Lie Group $\U(D_{\Gamma})$ and Lie Algebras $\u(D_{\Gamma})$ and $\u(N)_{\Gamma}$}
\label{ssec:U(D_Gamma)-u(N)}
Using $D_{\Gamma}$ defined in \eqref{eq:D_Gamma}, let us also define a Lie group
\begin{equation*}
  \U(D_{\Gamma}) \defeq \setdef{ U \in \C^{N\times N} }{ U D_{\Gamma} U^{*} = D_{\Gamma} },
\end{equation*}
and its (right) action on $\C^{2 \times N}$:
\begin{equation*}
  \U(D_{\Gamma}) \times \C^{2 \times N} \to \C^{2 \times N};
  \qquad
  (U, \Phi) \mapsto \Phi U.
\end{equation*}
Again, it is clear that this action leaves $\Theta$ invariant as well; see \eqref{eq:Theta-C2N}.

The Lie algebra of $\U(D_{\Gamma})$ is given by
\begin{equation*}
  \u(D_{\Gamma}) \defeq \setdef{ \tilde{\zeta} \in \C^{N\times N} }{ \tilde{\zeta} D_{\Gamma} + D_{\Gamma} \tilde{\zeta}^{*} = 0 }.
\end{equation*}
Then we have the following \textit{vector space} isomorphism between $\u(D_{\Gamma})$ and the Lie algebra $\u(N)$ of the unitary group $\U(N)$:
\begin{equation*}
  \u(D_{\Gamma}) \to \u(N);
  \qquad
  \tilde{\zeta} \mapsto \tilde{\zeta} D_{\Gamma} \eqdef \zeta.
\end{equation*}
Note that this is \textit{not} a Lie algebra isomorphism.
However, we may equip $\u(N)$ with the modified Lie bracket
\begin{equation}
  \label{eq:commutator-u(N)}
  [\xi,\eta]_{\Gamma} \defeq \xi D_{\Gamma}^{-1} \eta - \eta D_{\Gamma}^{-1} \xi
\end{equation}
to define a Lie algebra $\u(N)_{\Gamma}$.
Then the above vector space isomorphism becomes a Lie algebra isomorphism between $\u(D_{\Gamma})$ (with the standard commutator) and $\u(N)_{\Gamma}$ with the modified Lie bracket~\eqref{eq:commutator-u(N)}.

Let us equip $\u(N)_{\Gamma}$ with the inner product in the same form as in \eqref{eq:ip-u2}, and identify the dual $\u(N)_{\Gamma}^{*}$ with $\u(N)_{\Gamma}$ via this inner product; hence we may identify $\u(D_{\Gamma})^{*}$ with $\u(N)_{\Gamma}$ as well.
Under this identification, the adjoint and coadjoint representations of $\U(D_{\Gamma})$ on $\u(N)_{\Gamma}$ and $\u(N)_{\Gamma}^{*}$ are
\begin{equation}
  \label{eq:Ad_and_Adstar-U(D_Gamma)}
  \Ad_{U}\xi = U \xi U^{*},
  \qquad
  \Ad_{U}^{*}\lambda = U^{*} \lambda U,
\end{equation}
and also the corresponding $\u(D_{\Gamma})$-representations are 
\begin{equation}
  \label{eq:ad_and_adstar-u(N)}
  \ad_{\xi}\eta = [\xi,\eta]_{\Gamma},
  \qquad
  \ad_{\xi}^{*}\lambda = \lambda \xi D_{\Gamma}^{-1} - D_{\Gamma}^{-1} \xi \lambda
\end{equation}
for every $U \in \U(D_{\Gamma})$, $\xi, \eta \in \u(N)_{\Gamma}$, and $\lambda \in \u(N)_{\Gamma}^{*}$.

\begin{lemma}
  The momentum map $\mathbf{L}\colon \C^{2\times N} \to \u(N)_{\Gamma}^{*}$ associated with the above $\U(D_{\Gamma})$-action is
  \begin{equation}
    \label{eq:L}
    \mathbf{L}(\Phi) = -\frac{\rmi}{R}\Phi^{*} \Phi
    = -\frac{\rmi}{R}
    \begin{bNiceMatrix}
      \norm{ \varphi_{1} }^{2} & \varphi^{*}_{1} \varphi_{2} & \Cdots & \varphi^{*}_{1} \varphi_{N} \\
      \varphi^{*}_{2} \varphi_{1}  & \norm{ \varphi_{2} }^{2} & \Ddots & \Vdots \\
      \Vdots & \Ddots & \Ddots & \varphi^{*}_{N-1} \varphi_{N} \\
      \varphi^{*}_{N} \varphi_{1} & \Cdots & \varphi^{*}_{N} \varphi_{N-1} & \norm{ \varphi_{N} }^{2}
    \end{bNiceMatrix}.
  \end{equation}
\end{lemma}
\begin{proof}
  The infinitesimal generator of an arbitrary element $\tilde{\zeta} \in \u(D_{\Gamma})$ is
  \begin{equation*}
    \tilde{\zeta}_{\C^{2 \times N}}(\Phi) = \Phi\tilde{\zeta} = \Phi \zeta D_{\Gamma}^{-1}.
  \end{equation*}
  Hence the associated momentum map $\mathbf{L}\colon \C^{2 \times N} \to \u(N)_{\Gamma}^{*}$ satisfies 
  \begin{align*}
    \ip{ \mathbf{L}(\Phi) }{ \zeta }
    &= \Theta(\Phi) \cdot \tilde{\zeta}_{\C^{2 \times N}}(\Phi) \\
    &= -\frac{2}{R}\Im\parentheses{ \tr\parentheses{ D_{\Gamma} \Phi^{*} \Phi\zeta D_{\Gamma}^{-1} } } \\
    &= -\frac{2}{R}\Im\parentheses{ \tr\parentheses{ \Phi^{*} \Phi\zeta } } \\
    &= \frac{2}{R}\rmi \tr\parentheses{ \Phi^{*} \Phi\zeta } \\
    &= \frac{2}{R}\tr\parentheses{ (-\rmi\Phi^{*} \Phi)^{*}\zeta } \\
    &= \ip{ -\frac{\rmi}{R}\Phi^{*} \Phi }{ \zeta },
  \end{align*}
  where we used the fact that $\tr\parentheses{ \Phi^{*} \Phi\zeta }$ is pure imaginary.
  Hence the expression~\eqref{eq:L} follows.
\end{proof}

\subsection{$\U(2)$-Reduction via a Dual Pair}
\label{ssec:U(2)-reduction}

\begin{proposition}
  The Hamiltonian reduction of the lifted dynamics~\eqref{eq:lifted_dynamics} by the $\U(2)$-symmetry yields the Lie--Poisson dynamics
  \begin{equation}
    \label{eq:LP-u(N)^*}
    \dot{\lambda} = \ad_{\delta h/\delta\lambda}^{*}\lambda
  \end{equation}
  in $\u(N)_{\Gamma}^{*}$, where $h\colon \u(N)_{\Gamma}^{*} \to \R$ is defined as
  \begin{equation}
    \label{eq:collectiveH}
    h(\lambda) \defeq
    -\frac{1}{4\pi R^{2}} \sum_{1\le i < j \le N} \Gamma_{i} \Gamma_{j} \ln \parentheses{
      R^{2} \parentheses{
        \frac{1}{2}\left( \lambda_{i} + \lambda_{j} \right)^{2} - |\lambda_{ij}|^{2}
      }
    }.
  \end{equation}
\end{proposition}
\begin{proof}
  As alluded at the beginning of the section, the pair of momentum maps $\mathbf{K}$ and $\mathbf{L}$ form a dual pair in the sense of \citet{We1983}; see also \cite[Section~IV.7]{LiMa1987} and \cite[Chapter~11]{OrRa2004}.
  Specifically, the above $\U(2)$- and $\U(D_{\Gamma})$-actions along with the associated momentum maps $\mathbf{K}$ and $\mathbf{L}$ define so-called \textit{mutually transitive actions} on $\C^{2\times N}$ (see \citet{Sk2019} and \citet{SkVi2019}) in the following sense:
  (i)~The $\U(2)$-action and the $\U(D_{\Gamma})$-action commute; (ii)~they are symplectic actions; (iii)~the momentum maps $\mathbf{K}$ and $\mathbf{L}$ are equivariant; (iv)~each level set of $\mathbf{K}$ is a $\U(D_{\Gamma})$-orbit, and each level set of $\mathbf{L}$ is an $\U(2)$-orbit.
  In fact, this is essentially a special case of \citet[Section~3]{SkVi2019}.

  This implies the following (see, e.g., \cite[Proposition~2.8]{SkVi2019}):
  For every $\Phi_{0} \in \C^{2\times N}$, let $\kappa_{0} \defeq \mathbf{K}(\Phi_{0})$ and $\lambda_{0} \defeq \mathbf{L}(\Phi_{0})$; then the Marsden--Weinstein quotient $\mathbf{K}^{-1}(\kappa_{0})/\U(2)_{\kappa_{0}}$ is symplectomorphic to the coadjoint orbit $\mathcal{O}_{\lambda_{0}}$ passing through $\lambda_{0} \in \u(N)_{\Gamma}^{*}$, where $\mathcal{O}_{\lambda_{0}}$ is equipped with the $(-)$-Kirillov--Kostant--Souriau (KKS) symplectic structure (see, e.g., \citet[Chapter~1]{Ki2004} and \citet[Chapter~14]{MaRa1999} and references therein; note that the $\U(D_{\Gamma})$-action is a right action, and hence it is $(-)$-KKS): For every $\lambda \in \mathcal{O}_{\lambda_{0}} \subset \u(N)_{\Gamma}^{*}$ and $\xi, \eta \in \u(N)_{\Gamma}$,
  \begin{equation}
    \label{eq:KKS-u(N)}
    \Omega_{\mathcal{O}_{\lambda_{0}}}(\lambda)\parentheses{ -\ad_{\xi}^{*}\lambda, -\ad_{\eta}^{*}\lambda } \defeq -\ip{\lambda}{[\xi,\eta]_{\Gamma}}.
  \end{equation}

  This motivates us to set
  \begin{equation}
    \label{eq:lambda-Phi}
    \lambda
    = -\frac{\rmi}{2}
    \begin{bNiceMatrix}
      \sqrt{2}\lambda_{1} & \lambda_{12} & \Cdots & \lambda_{1N} \\
      \bar{\lambda}_{12}  & \sqrt{2}\lambda_{2} & \Ddots & \Vdots \\
      \Vdots & \Ddots & \Ddots & \lambda_{N-1,N} \\
      \bar{\lambda}_{1N} & \Cdots & \bar{\lambda}_{N-1,N} & \sqrt{2}\lambda_{N}
    \end{bNiceMatrix}
    = \mathbf{L}(\Phi),
  \end{equation}
  or in view of \eqref{eq:L},
  \begin{equation}
    \label{eq:lambda-varphi}
    \lambda_{i} = \frac{\sqrt{2}}{R}\, \norm{\varphi_{i}}^{2}
    \quad\text{for}\quad i = 1, \dots, N,
    \qquad
    \lambda_{ij} = \frac{2}{R} \varphi_{i}^{*} \varphi_{j}
    \quad\text{for}\quad 1\le i < j \le N.
  \end{equation}
  We also define a collective Hamiltonian~\cite{GuSt1980} $h$ so that $h \circ \mathbf{L} = H$ (see \eqref{eq:H-C^2} for an expression of $H$).
  Then, the reduced dynamics in the the Marsden--Weinstein quotient $\mathbf{K}^{-1}(\kappa_{0})/\U(2)_{\kappa_{0}}$ is equivalent to the Lie--Poisson dynamics~\eqref{eq:LP-u(N)^*} in $\u(N)_{\Gamma}^{*}$.
\end{proof}

\subsection{Lie--Poisson Bracket on $\u(N)_{\Gamma}^{*}$}
\label{ssec:LP-u(N)}
One may also write the Lie--Poisson equation~\eqref{eq:LP-u(N)^*} as
\begin{equation*}
  \dot{\lambda} = \PB{\lambda}{h},
\end{equation*}
where the Poisson bracket is the $(-)$-Lie--Poisson bracket on $\u(N)_{\Gamma}^{*}$ corresponding to the above symplectic form~\eqref{eq:KKS-u(N)}, i.e.,
\begin{equation}
  \label{eq:LPB-u(N)^*-0}
  \PB{f}{h}(\lambda) \defeq -\ip{\lambda}{ \brackets{\fd{f}{\lambda}, \fd{h}{\lambda}}_{\Gamma} }
\end{equation}
for all smooth $f, h \colon \u(N)_{\Gamma}^{*} \to \R$.
In this subsection, we would like to find a concrete expression for the bracket.

To that end, let us first define an orthonormal basis for $\u(N)_{\Gamma}$.
Let $e_{i} \in \R^{N}$ be the unit vector whose $i$-th component is 1, and define
\begin{gather*}
  \mathcal{D}_{i} \defeq -\frac{\rmi}{\sqrt{2}} e_{i}e_{i}^{T} = -\frac{\rmi}{\sqrt{2}} \diag(e_{i}) \quad\text{for}\quad i \in \{1, \dots, N\},\\
  \mathcal{E}_{ij} \defeq -\frac{\rmi}{2} \parentheses{ e_{i}e_{j}^{T} + e_{j}e_{i}^{T} }, \quad
  \mathcal{F}_{ij} \defeq \frac{1}{2} \parentheses{ e_{i}e_{j}^{T} - e_{j}e_{i}^{T} } \quad\text{for}\quad i,j \in \{1, \dots, N\}.
\end{gather*}
Note that $\mathcal{E}_{ii} = \sqrt{2}\mathcal{D}_{i}$ and $\mathcal{F}_{ii} = 0$ for $i \in \{1, \dots, N\}$.
One then sees that
\begin{equation*}
  \{ \mathcal{D}_{i} \}_{i=1}^{N} \cup \{ \mathcal{E}_{ij},\, \mathcal{F}_{ij} \}_{1\le i < j \le N}
\end{equation*}
forms a basis for $\u(N)_{\Gamma}$.
Hence we may write an arbitrary element $\xi \in \u(N)_{\Gamma}$ as follows:
\begin{align}
  \label{eq:xi-u(N)}
  \xi &= (\xi_{1}, \dots, \xi_{N}, \xi_{12}, \dots, \xi_{N-1,N}) \nonumber\\
      &= \sum_{i=1}^{N} \xi_{i} \mathcal{D}_{i}
        + \sum_{1\le i < j \le N} \parentheses{ (\Re\xi_{ij}) \mathcal{E}_{ij} + (\Im\xi_{ij}) \mathcal{F}_{ij} } \nonumber\\
      &= -\frac{\rmi}{2}
        \begin{bNiceMatrix}
          \sqrt{2}\xi_{1} & \xi_{12} & \Cdots & \xi_{1N} \\
          \bar{\xi}_{12}  & \sqrt{2}\xi_{2} & \Ddots & \Vdots \\
          \Vdots & \Ddots & \Ddots & \xi_{N-1,N} \\
          \bar{\xi}_{1N} & \Cdots & \bar{\xi}_{N-1,N} & \sqrt{2}\xi_{N}
        \end{bNiceMatrix}.
\end{align}
So we may identify $\u(N)_{\Gamma}$ with $\R^{N} \times \R^{N(N-1)} = \R^{N^{2}}$ as a vector space.

It is then straightforward calculations to see that the Lie bracket~\eqref{eq:commutator-u(N)} on $\u(N)_{\Gamma}$ satisfies the following for all $i, j, k, l \in \{1, \dots, N\}$:
\begin{gather*}
  [\mathcal{D}_{i}, \mathcal{E}_{jk}]_{\Gamma} = -\frac{\Gamma_{i}^{-1}}{\sqrt{2}}( \delta_{ij} \mathcal{F}_{ik} + \delta_{ik} \mathcal{F}_{ij} ),
  \qquad
  [\mathcal{D}_{i}, \mathcal{F}_{jk}]_{\Gamma} = \frac{\Gamma_{i}^{-1}}{\sqrt{2}}( \delta_{ij} \mathcal{E}_{ik} - \delta_{ik} \mathcal{E}_{ij} ),
  \\
  [\mathcal{E}_{ij}, \mathcal{E}_{kl}]_{\Gamma} = -\frac{1}{2}\parentheses{
    \Gamma_{i}^{-1}\parentheses{ \delta_{ik} \mathcal{F}_{jl} + \delta_{il} \mathcal{F}_{jk} }
    + \Gamma_{j}^{-1}\parentheses{ \delta_{jk} \mathcal{F}_{il} + \delta_{jl} \mathcal{F}_{ik} }
  },
  \\
  [\mathcal{F}_{ij}, \mathcal{F}_{kl}]_{\Gamma} = -\frac{1}{2}\parentheses{
    \Gamma_{i}^{-1}\parentheses{ \delta_{ik} \mathcal{F}_{jl} - \delta_{il} \mathcal{F}_{jk} }
    - \Gamma_{j}^{-1}\parentheses{ \delta_{jk} \mathcal{F}_{il} - \delta_{jl} \mathcal{F}_{ik} }
  },
  \\
  [\mathcal{E}_{ij}, \mathcal{F}_{kl}]_{\Gamma} = \frac{1}{2}\parentheses{
    \Gamma_{i}^{-1}\parentheses{ \delta_{ik} \mathcal{E}_{jl} - \delta_{il} \mathcal{E}_{jk} }
    + \Gamma_{j}^{-1}\parentheses{ \delta_{jk} \mathcal{E}_{il} - \delta_{jl} \mathcal{E}_{ik} }
  },
\end{gather*}
where we did \textit{not} assume Einstein's summation convention.
Note that the first two are in fact special cases of the third and the last ones, respectively, because $\mathcal{D}_{i} = \mathcal{E}_{ii}/\sqrt{2}$.

Using the coordinates for $\u(N)_{\Gamma}^{*} \cong \u(N)_{\Gamma}$ with respect to the above basis, we may write an arbitrary element $\lambda \in \u(N)_{\Gamma}^{*}$ using the coordinates $(\lambda_{1}, \dots, \lambda_{N}, \lambda_{12}, \dots, \lambda_{N-1,N})$ just as we did in \eqref{eq:xi-u(N)} for $\xi \in \u(N)_{\Gamma}$.
Then we may express the Lie--Poisson bracket~\eqref{eq:LPB-u(N)^*-0} as follows: For all $i, j, k, l \in \{1, \dots, N\}$,
\begin{equation}
  \label{eq:LPB-u(N)^*}
  \begin{array}{c}
    \DS \PB{ \lambda_{i} }{ \lambda_{j} } = 0,
    \qquad
    \DS \PB{ \lambda_{i} }{ \lambda_{jk} } = -\rmi\,\frac{\Gamma_{i}^{-1}}{\sqrt{2}}(\delta_{ij} \lambda_{ik} - \delta_{ik} \lambda_{ji}),
    \smallskip\\
    \DS \PB{ \lambda_{ij} }{ \lambda_{kl} } = \rmi \parentheses{ \Gamma_{i}^{-1} \delta_{il} \lambda_{kj} - \Gamma_{j}^{-1} \delta_{jk} \lambda_{il} }.
  \end{array}
\end{equation}

\begin{remark}
  \label{rem:lambda_{ji}}
  As one can see in \eqref{eq:xi-u(N)}, we do not use entries $\lambda_{ij}$ with $i \ge j$ explicitly as coordinates in $\u(N)_{\Gamma}$ or $\u(N)_{\Gamma}^{*}$, but such entries may appear in the above Poisson bracket formulas.
  However, one may \textit{define} $\lambda_{ij} \defeq \ip{ \lambda }{ \mathcal{E}_{ij} } + \rmi \ip{ \lambda }{ \mathcal{F}_{ij} }$ even if $i \ge j$.
  Then it follows that $\lambda_{ij} = \bar{\lambda}_{ji}$ if $i > j$ as well as that $\lambda_{ii} = \sqrt{2} \lambda_{i}$.
  So we may rewrite the above Poisson bracket formulas in terms of the coordinates for $\u(N)_{\Gamma}^{*}$.
\end{remark}

The above Lie--Poisson bracket has the following family of Casimirs:
\begin{proposition}
  \label{prop:Casimirs}
  \leavevmode
  \begin{enumerate}[(i)]
  \item For every $j \in \N$, the function $C_{j}\colon \u(N)_{\Gamma}^{*} \to \R$ defined by
    \begin{equation*}
      C_{j}(\lambda) \defeq \tr\parentheses{ (\rmi D_{\Gamma} \lambda)^{j} }
    \end{equation*}
    is a Casimir function for the Lie--Poisson bracket~\eqref{eq:LPB-u(N)^*}.
    \label{prop:Casimirs-1}
    \medskip
  \item Those Casimirs $C_{j}$ with $j \ge N$ can be expressed in terms of $\{ C_{j} \}_{j=1}^{N-1}$.
    \label{prop:Casimirs-2}
  \end{enumerate}
\end{proposition}
\begin{proof}
  See \Cref{sec:Proof-Casimirs}.
\end{proof}

\section{Further Reduction by $\mathbb{T}^{N-1}$-symmetry}
\label{sec:further_reduction}
Let us now move on to the $\mathbb{T}^{N-1}$-reduction part in \Cref{fig:lift_and_reduce}.
Recall from \Cref{sec:lift} that the lifted dynamics picked up $\mathbb{T}^{N}$-symmetry.
We would like to get rid of this extra symmetry.

\subsection{$\mathbb{T}^{N-1}$-symmetry}
Consider the action
\begin{equation*}
  \mathbb{T}^{N-1} \times \u(N)_{\Gamma}^{*} \to \u(N)_{\Gamma}^{*}
\end{equation*}
defined by
\begin{multline}
  \label{eq:T^{N-1}-action}
  \parentheses{ (e^{\rmi\theta_{1}}, \dots, e^{\rmi\theta_{N-1}}), \lambda}
  \mapsto
  \Ad_{e^{-\rmi\tilde{\boldsymbol{\theta}}}}^{*} \lambda
  = e^{\rmi\tilde{\boldsymbol{\theta}}} \lambda e^{-\rmi\tilde{\boldsymbol{\theta}}} \\
  \footnotesize
  = -\frac{\rmi}{2}
  \begin{bNiceMatrix}
    \sqrt{2}\lambda_{1} & e^{\rmi(\theta_{1}-\theta_{2})}\lambda_{12} & \Cdots & e^{\rmi(\theta_{1}-\theta_{N-1})}\lambda_{1,N-1} & e^{\rmi\theta_{1}}\lambda_{1N} \smallskip\\
    e^{\rmi(\theta_{2}-\theta_{1})} \bar{\lambda}_{12} & \sqrt{2}\lambda_{2} & \Ddots & \Vdots & \Vdots \smallskip\\
    \Vdots & \Ddots & \Ddots & e^{\rmi(\theta_{N-2}-\theta_{N-1})}\lambda_{N-2,N-1} & \Vdots \smallskip\\
    e^{\rmi(\theta_{N-1}-\theta_{1})}\bar{\lambda}_{1,N-1} & \Cdots & e^{\rmi(\theta_{N-1}-\theta_{N-2})} \bar{\lambda}_{N-2,N-1} & \sqrt{2} \lambda_{N-1} & e^{\rmi\theta_{N-1}}\lambda_{N-1,N} \smallskip\\
    e^{-\rmi\theta_{1}} \bar{\lambda}_{1N} & \Cdots & e^{-\rmi\theta_{2}}\bar{\lambda}_{N-2,N} & e^{-\rmi\theta_{N-1}} \bar{\lambda}_{N-1,N} & \sqrt{2}\lambda_{N}
  \end{bNiceMatrix},
\end{multline}
where
\begin{equation*}
  e^{\rmi\tilde{\boldsymbol{\theta}}} \defeq \diag\parentheses{ e^{\rmi\theta_{1}}, \dots, e^{\rmi\theta_{N-1}}, 1 } \in \U(D_{\Gamma}).
\end{equation*}
Note that this action restricts to the coadjoint orbits because it is a coadjoint action by elements in $\U(D_{\Gamma})$; see \eqref{eq:Ad_and_Adstar-U(D_Gamma)}.

\begin{remark}
  \label{rem:U(2)-T^{N-1}}
  Why do we consider the $\mathbb{T}^{N-1}$-action instead of the more natural $\mathbb{T}^{N}$-action?
  It is because the above $\U(2)$-symmetry took into account an $\mathbb{S}^{1}$-symmetry out of the $\mathbb{T}^{N}$-symmetry already.
  This is the ``global'' $\mathbb{S}^{1}$-phase symmetry alluded in \Cref{rem:U(2)-SU(2)}:
  A part of the $\mathbb{T}^{N}$-symmetry is the invariance under the $\mathbb{S}^{1}$-action that changes the phase of the entire system by the same amount $\Phi \mapsto e^{\rmi\theta} \Phi$, but this is an $\mathbb{S}^{1}$ subgroup action of \eqref{eq:U(2)-action}.
  The above $\mathbb{T}^{N-1}$-action takes care of the rest of the $\mathbb{T}^{N}$-symmetry the lifted dynamics picked up.
\end{remark}

Clearly the symplectic structure~\eqref{eq:KKS-u(N)} and the collective Hamiltonian~\eqref{eq:collectiveH} are invariant under the above $\mathbb{T}^{N-1}$-action, and thus the $\U(2)$-reduced dynamics~\eqref{eq:LP-u(N)^*} has the $\mathbb{T}^{N-1}$-symmetry.
Let us find the associated momentum map.
Let
\begin{equation*}
  \omega = (\omega_{1}, \dots, \omega_{N-1},0) \in (T_{1}\mathbb{S}^{1})^{N-1} \cong \R^{N-1},  
\end{equation*}
and define
\begin{equation*}
  \omega_{\Gamma}
  \defeq
  (\Gamma_{1}\omega_{1}, \dots, \Gamma_{N-1}\omega_{N-1}, 0)
\end{equation*}
so that
\begin{equation*}
  D_{\omega} = \diag(\omega_{1}, \dots, \omega_{N-1}, 0),
  \qquad
  D_{\omega_{\Gamma}} = \diag(\Gamma_{1}\omega_{1}, \dots, \Gamma_{N-1}\omega_{N-1}, 0)
  = D_{\Gamma} D_{\omega} = D_{\omega} D_{\Gamma}.
\end{equation*}
Then the infinitesimal generator corresponding to $\omega$ is
\begin{align*}
  \omega_{\u(N)_{\Gamma}^{*}}(\lambda)
  &= \left.\od{}{s} \Ad_{e^{-\rmi s\omega}}^{*} \lambda \right|_{s=0} \\
  &= [\rmi D_{\omega}, \lambda] \\
  &= [\lambda, -\rmi D_{\omega}] \\
  &= \lambda (-\rmi D_{\omega}) D_{\Gamma} D_{\Gamma}^{-1} - D_{\Gamma}^{-1} D_{\Gamma} (-\rmi D_{\omega}) \lambda \\
  &= \lambda (-\rmi D_{\omega_{\Gamma}}) D_{\Gamma}^{-1} - D_{\Gamma}^{-1} (-\rmi D_{\omega_{\Gamma}}) \lambda \\
  &= \ad_{-\rmi D_{\omega_{\Gamma}}}^{*}\lambda,
\end{align*}
where we used the expression for $\ad^{*}$ in \eqref{eq:ad_and_adstar-u(N)}.
Then we see that $\omega_{\u(N)_{\Gamma}^{*}}(\lambda) = \ad_{\delta \mathcal{N}^{\omega}/\delta\lambda}^{*} \lambda$ with $\mathcal{N}^{\omega}\colon \u(N)_{\Gamma}^{*} \to \R$ defined by
\begin{equation*}
  \mathcal{N}^{\omega}(\lambda)
  = \ip{\lambda}{-\rmi D_{\omega_{\Gamma}}}
  = -2\rmi\tr(\lambda^{*} D_{\omega_{\Gamma}}) 
  = \sqrt{2}\,\sum_{i=1}^{N-1} \Gamma_{i} \lambda_{i} \omega_{i}.
\end{equation*}
The associated momentum map $\mathbf{N}\colon \u(N)_{\Gamma}^{*} \to (\R^{N-1})^{*} \cong \R^{N-1}$ then satisfies $N^{\omega}(\lambda) = \mathbf{N}(\lambda) \cdot \omega_{\Gamma}$, and thus we obtain
\begin{equation*}
  \mathbf{N}(\lambda) = \sqrt{2}\,(\Gamma_{1}\lambda_{1}, \dots, \Gamma_{N-1}\lambda_{N-1}).
\end{equation*}
These turn out to be trivial invariants in our setting because, in view of \eqref{eq:lambda-varphi} and using $\varphi_{i} \in \mathbb{S}^{3}_{\sqrt{R}}$, we have $\lambda_{i} = \sqrt{2}$ for every $i \in \{1, \dots, N\}$.

Note that the action~\eqref{eq:T^{N-1}-action} is not free.
However, if we restrict the action to the open subset
\begin{equation*}
  \mathring{\u}(N)_{\Gamma}^{*} \defeq \setdef{ \lambda \in \u(N)_{\Gamma}^{*} }{ \lambda_{ij} \neq 0 \text{ for all $i,j \in \{1, \dots, N\}$ with $i \neq j$ } },
\end{equation*}
then the action becomes free.
Its geometric interpretation is the following: If, for example, $\lambda_{12} = 0$ then $\varphi_{1}^{*}\varphi_{2} = 0$, and this along with \eqref{eq:S^3_distance-S^2} implies $| \mathbf{x}_{1} - \mathbf{x}_{2} | = 2R$, i.e., vortices 1 and 2 are in the antipodal points.
In this case, there is no well-defined geodesic connecting the two vortices on the sphere, and hence the ``shape'' of the vortices on the sphere is not well-defined.

\subsection{Reduction by $\mathbb{T}^{N-1}$-symmetry}
Let us define
\begin{equation*}
  \mu_{ijk} \defeq \lambda_{ij} \lambda_{ki} \lambda_{jk} \in \mathring{\C}
  \quad\text{with}\quad
  \mathring{\C} \defeq \C\backslash\{0\},
\end{equation*}
and, also as a shorthand,
\begin{equation}
  \label{eq:mu_ij}
  \mu_{ij} \defeq \mu_{ijN} = \lambda_{ij} \bar{\lambda}_{iN} \lambda_{jN} = \lambda_{ij} \lambda_{Ni} \lambda_{jN} \in \mathring{\C}
\end{equation}
for all $i, j \in \{1, \dots, N-1\}$ with $i < j$.
These variables provide an alternative parametrization of the entries $\{ \lambda_{ij} \}_{1 \le i < j \le N-1}$ of $\lambda$, i.e., those $(N-1)(N-2)/2$ entries of $\lambda$ in \eqref{eq:lambda-Phi} that are above the main diagonal except those in the last column.

Therefore, we may parametrize $\lambda \in \mathring{\u}(N)_{\Gamma}^{*}$ as follows:
\begin{equation*}
  \lambda = (
  \lambda_{1}, \dots, \lambda_{N},
  \lambda_{1N}, \dots, \lambda_{N-1,N},
  \mu_{12}, \dots, \mu_{N-2,N-1}
  )
  \in \R^{N} \times \mathring{\C}^{N-1} \times \mathring{\C}^{(N-1)(N-2)/2}.
\end{equation*}
Then the $\mathbb{T}^{N-1}$-action~\eqref{eq:T^{N-1}-action} becomes trivial on the variables $\{ \mu_{ij} \}_{1 \le i < j \le N-1}$, and hence we have
\begin{align*}
  \mathring{\u}(N)_{\Gamma}^{*}/\mathbb{T}^{N-1}
  &= \R^{N} \times \parentheses{ \mathring{\C}^{N-1}/\mathbb{T}^{N-1} } \times \mathring{\C}^{(N-1)(N-2)/2} \\
  &= \R^{N} \times \R_{+}^{N-1} \times \mathring{\C}^{(N-1)(N-2)/2} \\
  &= \braces{ \parentheses{
    \lambda_{1}, \dots, \lambda_{N},
    |\lambda_{1N}|, \dots, |\lambda_{N-1,N}|,
    \mu_{12}, \dots, \mu_{N-2,N-1}
    } }.
\end{align*}
Then the Poisson bracket on $\mathring{\u}(N)_{\Gamma}^{*}$ drops to the quotient by the standard Poisson reduction; see, e.g., \citet[Theorem~10.5.1]{MaRa1999}.
However, we may disregard $(\lambda_{1}, \dots, \lambda_{N})$ from the variables because
\begin{equation*}
  \lambda_{i} = \frac{\sqrt{2}}{R}\norm{\varphi_{i}}^{2} = \sqrt{2}
  \quad\text{for}\quad
  i = 1, \dots, N.
\end{equation*}
Also, since we have $|\lambda_{ij}|^{2} = 4 - (\ell_{ij}/R)^{2}$, we impose
\begin{equation*}
  0 < \ell_{ij} < 2R \iff 0 < |\lambda_{ij}| < 2
\end{equation*}
to avoid collisions and having vortices at antipodal points.
As a result, we have the following parametrization for the shape dynamics of $N$ point vortices:
\begin{align*}
  \mathcal{S}_{N} &\defeq
  (0,2)^{N-1} \times \mathring{\C}^{(N-1)(N-2)/2} \\
  &= \{ (
    |\lambda_{1N}|, \dots, |\lambda_{N-1,N}|,
    \mu_{12}, \dots, \mu_{N-2,N-1}
    ) \eqdef \zeta \},
\end{align*}
Note that the dimension of this manifold is $(N - 1)^{2}$, whereas the number of the ``internal'' variables $\{ \ell_{ij} \}_{1\le i < j \le N} \cup \{ V_{ijk} \}_{1\le i < j < k \le N}$ in \citet{BoPa1998} is $N(N^{2} - 1)/6$.

One can also show that
\begin{align*}
  \Re\mu_{ijk}
  &= \Re(\lambda_{ij} \lambda_{ki} \lambda_{jk}) \\
  &= \frac{8}{R^{3}}\Re\parentheses{ (\varphi_{i}^{*}\varphi_{j}) (\varphi_{k}^{*}\varphi_{i}) (\varphi_{j}^{*}\varphi_{k}) } \\
  &= \frac{4}{R^{2}} \parentheses{
    \abs{ \varphi_{i}^{*}\varphi_{j} }^{2}
    + \abs{ \varphi_{k}^{*}\varphi_{i} }^{2}
    + \abs{ \varphi_{j}^{*}\varphi_{k} }^{2} - R^{2}
    } \\
  &= |\lambda_{ij}|^{2} + |\lambda_{ki}|^{2} + |\lambda_{jk}|^{2} - 4,
\end{align*}
and
\begin{equation*}
  \Im\mu_{ijk} = \frac{2}{R^{3}} V_{ijk}
  \quad\text{with}\quad
  V_{ijk} \defeq \mathbf{x}_{i} \cdot (\mathbf{x}_{j} \times \mathbf{x}_{k}).
\end{equation*}
Thus, using \eqref{eq:mu_ij}, 
\begin{equation}
  \label{eq:Remu_ij}
  \Re\mu_{ij} = \Re\mu_{ijN}
  = \frac{ (\Re\mu_{ij})^{2} + (\Im\mu_{ij})^{2} }{ |\lambda_{iN}|^{2} |\lambda_{jN}|^{2} } + |\lambda_{iN}|^{2} + |\lambda_{jN}|^{2} - 4,
\end{equation}
and hence $\Re\mu_{ij}$ implicitly depends on $|\lambda_{iN}|$, $|\lambda_{jN}|$, and $\Im \mu_{ij}$.
Therefore, our shape dynamics is effectively defined on the $N(N - 1)/2$-dimensional manifold
\begin{equation*}
  (0,2)^{N-1} \times \R^{(N-1)(N-2)/2}
  = \{(
  |\lambda_{1N}|, \dots, |\lambda_{N-1,N}|,
  \Im\mu_{12}, \dots, \Im\mu_{N-2,N-1}
  ) \}.
\end{equation*}
However, practically speaking, it is simpler to retain $\{ \Re\mu_{ij} \}_{1\le i < j \le N-1}$ as independent variables and impose \eqref{eq:Remu_ij} as constraints instead.
In other words, we may define functions
\begin{equation}
  \label{eq:f_ij}
  f_{ij}\colon \mathcal{S}_{N} \to \R;
  \qquad
  f_{ij}(\zeta) \defeq \Re\mu_{ij} - \frac{ (\Re\mu_{ij})^{2} + (\Im\mu_{ij})^{2} }{ |\lambda_{iN}|^{2} |\lambda_{jN}|^{2} } - |\lambda_{iN}|^{2} - |\lambda_{jN}|^{2} + 4
\end{equation}
with $1\le i < j \le N-1$.
Then the shape dynamics is on the level set $\bigcap_{1\le i < j \le N-1} f_{ij}^{-1}(0)$.

\subsection{Shape Variables}
How do the variables for $\mathcal{S}_{N}$ determine the ``shape'' formed by the vortices?
In view of \eqref{eq:S^3_distance-S^2}, \eqref{eq:Im(mu_{123})/8}, and \eqref{eq:lambda-varphi}, we may relate our variables with the inter-vortex distances $\ell_{ij} \defeq | \mathbf{x}_{i} - \mathbf{x}_{j} |$ and the signed volume $V_{ijk} = \mathbf{x}_{i} \cdot (\mathbf{x}_{j} \times \mathbf{x}_{k})$ as follows:
\begin{equation}
  \label{eq:shape_variables}
  \begin{array}{c}
    \DS \ell_{ij}^{2} = 4\parentheses{ R^{2} - |\varphi_{i}^{*}\varphi_{j}|^{2} }
    = R^{2}(4 - |\lambda_{ij}|^{2})
    \iff
    |\lambda_{ij}|^{2} = 4 - \frac{\ell_{ij}^{2}}{R^{2}},
    \\
    \DS V_{ijk} = 4\Im\parentheses{ (\varphi_{i}^{*}\varphi_{j}) (\varphi_{k}^{*}\varphi_{i}) (\varphi_{j}^{*}\varphi_{k}) }
    = \frac{R^{3}}{2}\Im \mu_{ijk}
    \iff
    \Im \mu_{ijk} = \frac{2}{R^{3}} V_{ijk}.
  \end{array}
\end{equation}
Therefore, $|\lambda_{ik}|$ and $|\lambda_{jk}|$ specify the lengths $\ell_{ik}$ and $\ell_{jk}$ of two edges of the triangle formed by the vortices $\{i,j,k\}$, whereas $\Im\mu_{ijk}$ specifies the angle between the edges, thereby determining the triangle formed by the vortices $\{i,j,k\}$.
Note also that the above equations~\eqref{eq:shape_variables} give the relationship between our variables and the following ``internal'' variables of \cite{BoPa1998}:
\begin{equation*}
  \{ M_{ij} \defeq \ell_{ij}^{2} \}_{1\le i < j \le N} \cup \{ \Delta_{ijk} \defeq V_{ijk} \}_{1\le i < j < k \le N}.
\end{equation*}

Using the Lie--Poisson bracket~\eqref{eq:LPB-u(N)^*}, we find the Poisson structure on the space of our shape variables as follows:
For $1 \le i, j, k, l \le N-1$,
\begin{equation*}
  \PB{ |\lambda_{ij}|^{2} }{ |\lambda_{kl}|^{2} }
  = 2 \parentheses{
    \parentheses{ \frac{\delta_{ik}}{\Gamma_{i}} - \frac{\delta_{jk}}{\Gamma_{j}} } \Im\mu_{ijl}
    + \parentheses{ \frac{\delta_{il}}{\Gamma_{i}} - \frac{\delta_{jl}}{\Gamma_{j}} } \Im\mu_{ijk}
  }.
\end{equation*}
This along with \eqref{eq:shape_variables} recovers (2.11) of \citet{BoPa1998}\footnote{Note that our Poisson bracket~\eqref{eq:PB-R3} for the original dynamics is $R^{2}$ times their bracket and hence the difference by the factor $R^{2}$ carries over here as well.}.
In terms of our variables, we have, again using \eqref{eq:LPB-u(N)^*},
\begin{subequations}
  \label{eq:PB-shape_dynamics}
  \begin{equation}
    \PB{ |\lambda_{iN}|^{2} }{ |\lambda_{kN}|^{2} } = \frac{2}{\Gamma_{N}} \Im\mu_{ik}
    \quad\text{for}\quad
    1 \le i < k \le N-1.
  \end{equation}
  We also have, again using \eqref{eq:LPB-u(N)^*}, what correspond to (2.13) and (2.14) of \cite{BoPa1998} as follows:
  \begin{equation}
    \PB{ |\lambda_{iN}|^{2} }{ \mu_{kl} } =
    \begin{cases}
      \DS \rmi\parentheses{
        \frac{|\mu_{il}|^{2}}{\Gamma_{N} |\lambda_{lN}|^{2}} - \frac{|\lambda_{iN}|^{2} |\lambda_{lN}|^{2}}{\Gamma_{i}}
        - 2\parentheses{ \frac{1}{\Gamma_{N}} - \frac{1}{\Gamma_{i}} } \mu_{il}
      }
      % & i = k \text{ and } i \neq l, \medskip\\
      % \DS \rmi\parentheses{
      % \frac{|\lambda_{iN}|^{2} |\lambda_{kN}|^{2}}{\Gamma_{i}} - \frac{|\mu_{ki}|^{2}}{\Gamma_{N} |\lambda_{kN}|^{2}}
      % - 2\parentheses{ \frac{1}{\Gamma_{i}} - \frac{1}{\Gamma_{N}} } \mu_{ki}
      % }
      & i \neq k \text{ and } i = l, \medskip\\
      \DS \rmi\,\frac{\mu_{kl}}{\Gamma_{N}}
      \parentheses{
        \frac{\mu_{li}}{|\lambda_{lN}|^{2}} - \frac{\mu_{ik}}{|\lambda_{kN}|^{2}}
      }
      & i \neq k \text{ and } i \neq l,
    \end{cases}
  \end{equation}
  as well as
  \begin{multline}
    \PB{ \mu_{ij} }{ \mu_{lm} }
    \\
    = \begin{cases}
      \DS \rmi\parentheses{
        \frac{1}{\Gamma_{i}}
        \parentheses{ |\lambda_{jN}|^{2} \mu_{im} - |\lambda_{mN}|^{2} \mu_{ij} }
        + \frac{1}{\Gamma_{N}|\lambda_{iN}|^{2}}
        \parentheses{
          \frac{|\mu_{im}|^{2} \mu_{ij}}{ |\lambda_{mN}|^{2} } - \frac{|\mu_{ij}|^{2} \mu_{im}}{ |\lambda_{jN}|^{2} }
        }
      }
      & i = l, j \neq m, \medskip\\
      \DS \rmi\parentheses{
        \frac{1}{\Gamma_{j}}
        \parentheses{ |\lambda_{lN}|^{2} \mu_{ij} - |\lambda_{iN}|^{2} \mu_{lj} }
        + \frac{1}{\Gamma_{N}|\lambda_{jN}|^{2}}
        \parentheses{
          \frac{|\mu_{ij}|^{2} \mu_{lj}}{ |\lambda_{iN}|^{2} } - \frac{|\mu_{jl}|^{2} \mu_{ij}}{ |\lambda_{lN}|^{2} }
        }
      }
      & i \neq l, j = m, \medskip\\
      \DS \rmi \parentheses{
        2\parentheses{ \frac{1}{\Gamma_{j}} - \frac{1}{\Gamma_{N}} }
        \frac{\mu_{ij} \mu_{jm}}{|\lambda_{jN}|^{2}}
        - \frac{|\lambda_{jN}|^{2}}{\Gamma_{j}}\mu_{im}
        - \frac{ \mu_{ij} \mu_{mi} \mu_{jm} }{ \Gamma_{N} |\lambda_{iN}|^{2} |\lambda_{mN}|^{2} }
      }
      & i < j = l < m.
      % 
      % \DS \rmi \parentheses{ \frac{1}{\Gamma_{N}} - \frac{1}{\Gamma_{i}} } \parentheses{
      % \frac{2\mu_{ij} \mu_{li}}{|\lambda_{iN}|^{2}}
      % - |\lambda_{iN}|^{2} (\mu_{lj} - \mu_{ijl})
      % }
      %   & l < m = i < j.
    \end{cases}
  \end{multline}
\end{subequations}

Rewriting the collective Hamiltonian~\eqref{eq:collectiveH} in terms of our variables, we have the Hamiltonian $\mathcal{H}\colon \mathcal{S}_{N} \to \R$ defined as
\begin{multline}
  \label{eq:mathcalH}
  \mathcal{H}(\zeta)
  \defeq -\frac{1}{4\pi R^{2}} \Biggl(
  \sum_{1\le i < j \le N-1} \Gamma_{i} \Gamma_{j} \ln \parentheses{
    R^{2} \parentheses{ 4 - \frac{ |\mu_{ij}|^{2} }{ |\lambda_{iN}|^{2} |\lambda_{jN}|^{2} } }
  } \\
  + \Gamma_{N} \sum_{1\le i \le N-1} \Gamma_{i} \ln \parentheses{
    R^{2} \parentheses{ 4 - |\lambda_{iN}|^{2} }
  }
  \Biggr).
\end{multline}

\begin{remark}
  The terms inside the logarithmic functions in the above Hamiltonian are positive if we exclude collisions and antipodal configurations.
  In fact, using the definition~\eqref{eq:mu_ij} of $\mu_{ij}$,
  \begin{equation*}
    |\mu_{ij}| = |\lambda_{ij}|\, |\lambda_{iN}|\, |\lambda_{jN}| < 2 |\lambda_{iN}|\, |\lambda_{jN}|
    \iff
    \frac{ |\mu_{ij}| }{ |\lambda_{iN}| |\lambda_{jN}| } < 2
  \end{equation*}
  because we impose $| \lambda_{ij} | \in (0,2)$ to avoid collisions and antipodal configurations of vortices $i$ and $j$.
  So $|\mu_{ij}|$ cannot take an arbitrary value in $\mathring{\C}$.
  However, we retain the definition of $\mathcal{S}_{N}$ as is for simplicity.
\end{remark}

Let us summarize our main result as follows:
\begin{theorem}
  \label{thm:main}
  \leavevmode
  \begin{enumerate}[(i)]
  \item The shape dynamics of $N$ point vortices on the sphere is the Hamiltonian dynamics on
    \begin{equation*}
      \mathcal{S}_{N} \defeq (0,2)^{N-1} \times \mathring{\C}^{(N-1)(N-2)/2} = \braces{ \parentheses{
          |\lambda_{1N}|, \dots, |\lambda_{N-1,N}|,
          \mu_{12}, \dots, \mu_{N-2,N-1}
        } \eqdef \zeta }
    \end{equation*}
    with respect to the Poisson bracket~\eqref{eq:PB-shape_dynamics} and the Hamiltonian~\eqref{eq:mathcalH}.
    \label{thm:main-1}    
    \medskip
  \item The level set $\DS\bigcap_{1\le i < j \le N-1} f_{ij}^{-1}(0)$ with $f_{ij}$ defined in \eqref{eq:f_ij} is an invariant manifold of the shape dynamics.
    \label{thm:main-f}
    \medskip
  \item The Casimirs $\{ C_{j} \}_{j\in\N}$ from \Cref{prop:Casimirs} are invariants of the shape dynamics.
    \label{thm:main-Casimirs}
  \end{enumerate}
\end{theorem}

\begin{example}[The 3-Vortex Case: $N = 3$]
  The shape variables are given by
  \begin{equation*}
    \mathcal{S}_{3} \defeq
    (0,2)^{2} \times \mathring{\C} \\
    = \{ ( |\lambda_{13}|, |\lambda_{23}|, \mu_{12} ) \eqdef \zeta \}.
  \end{equation*}
  Here the shape of the vortices is given by a single triangle.
  As one can see from \eqref{eq:shape_variables}, the variables $|\lambda_{13}|$ and $|\lambda_{23}|$ specify the lengths $\ell_{13}$ and $\ell_{23}$ of the two edges at the vertex given by vortex 3, and $\Im\mu_{12}$ specifies the angle between the edges, hence determining the shape of the triangle.

  Therefore, $\Re\mu_{12}$ is redundant as a shape variable.
  Indeed it is implicitly defined in terms of $(|\lambda_{13}|, |\lambda_{23}|, \Im \mu_{12})$ because the level set $f_{12}^{-1}(0)$ with
  \begin{equation}
    \label{eq:Remu_12}
    f_{12}(\zeta) \defeq \Re\mu_{12}
    - \frac{ (\Re\mu_{12})^{2} + (\Im\mu_{12})^{2} }{ |\lambda_{13}|^{2} |\lambda_{23}|^{2} }
    - |\lambda_{13}|^{2} - |\lambda_{23}|^{2} + 4
  \end{equation}
  is an invariant manifold of the shape dynamics.
  
  The Poisson bracket is given by
  \begin{gather*}
    \PB{ |\lambda_{13}|^{2} }{ |\lambda_{23}|^{2} } = \frac{2}{\Gamma_{3}}\Im \mu_{12}, \\
    \PB{ |\lambda_{13}|^{2} }{ \mu_{12} }
    = \rmi \parentheses{
      |\lambda_{13}|^{2} \parentheses{ \frac{|\lambda_{12}|^{2}}{\Gamma_{3}} - \frac{|\lambda_{23}|^{2}}{\Gamma_{1}} }
      - 2\parentheses{ \frac{1}{\Gamma_{3}} - \frac{1}{\Gamma_{1}} } \mu_{12}
    }, \\
    \PB{ |\lambda_{23}|^{2} }{ \mu_{12} }
    = \rmi \parentheses{
      |\lambda_{23}|^{2} \parentheses{ \frac{|\lambda_{13}|^{2}}{\Gamma_{2}} - \frac{|\lambda_{12}|^{2}}{\Gamma_{3}} }
      - 2\parentheses{ \frac{1}{\Gamma_{2}} - \frac{1}{\Gamma_{3}} } \mu_{12}
    },
  \end{gather*}
  and the collective Hamiltonian~\eqref{eq:mathcalH} is
  \begin{equation*}
    \mathcal{H}(\zeta)
    = -\frac{1}{4\pi R^{2}} \parentheses{
      \Gamma_{1} \Gamma_{2} \ln \parentheses{
        R^{2}\parentheses{
          4 - \frac{ |\mu_{12}|^{2} }{ |\lambda_{13}|^{2} |\lambda_{23}|^{2} }
        }
      }
      + \Gamma_{3} \sum_{1\le i \le 2} \Gamma_{i} \ln \parentheses{
        R^{2}(4 - |\lambda_{i3}|^{2})
      }
    }.
  \end{equation*}
  The shape dynamics is then described as the Hamiltonian dynamics using the above Poisson bracket and this Hamiltonian.

  Let us find expressions for the Casimirs from \Cref{prop:Casimirs}.
  The first Casimir
  \begin{equation*}
    C_{1} = \frac{1}{\sqrt{2}}\sum_{i=1}^{3}\Gamma_{i} \lambda_{i} = \sum_{i=1}^{3}\Gamma_{i}
  \end{equation*}
  is a trivial one, and the second one
  \begin{equation*}
    C_{2} = \sum_{i=1}^{3}\Gamma_{i}^{2} + \frac{1}{2} \sum_{1 \le i < j \le 3} \Gamma_{i} \Gamma_{j} \abs{\lambda_{ij}}^{2}
    = \sum_{i=1}^{3}\Gamma_{i}^{2} + \frac{1}{2} \sum_{1 \le i < j \le 3} \Gamma_{i} \Gamma_{j} \parentheses{ 4 - \frac{\ell_{ij}^{2}}{R^{2}} }
  \end{equation*}
  is essentially the well-known invariant (see, e.g., \citet[Eq.~(4.2.6)]{Ne2001})
  \begin{equation*}
    \sum_{1\le i < j \le 3} \Gamma_{i} \Gamma_{j}\, \ell_{ij}^{2}.
  \end{equation*}
  Note that \Cref{prop:Casimirs}~\eqref{prop:Casimirs-2} implies that those Casimirs $C_{j}$ with $j \ge 3$ are not independent of $C_{1}$ and $C_{2}$.
\end{example}

\section{Application}
\label{sec:application}
\subsection{Tetrahedron Relative Equilibria}
Let us consider the special case with $N = 4$.
The shape variables in this case are
\begin{equation*}
  \{ \zeta = ( |\lambda_{14}|, |\lambda_{24}|, |\lambda_{34}|, \mu_{12}, \mu_{13}, \mu_{23} ) \} \in \mathcal{S}_{4} = (0,2)^{3} \times \mathring{\C}^{3}.
\end{equation*}
We are particularly interested in the stability of the tetrahedron relative equilibrium as shown in \Cref{fig:tetrahedron}.
\begin{figure}[htbp]
  \centering
  \includegraphics[width=.275\linewidth]{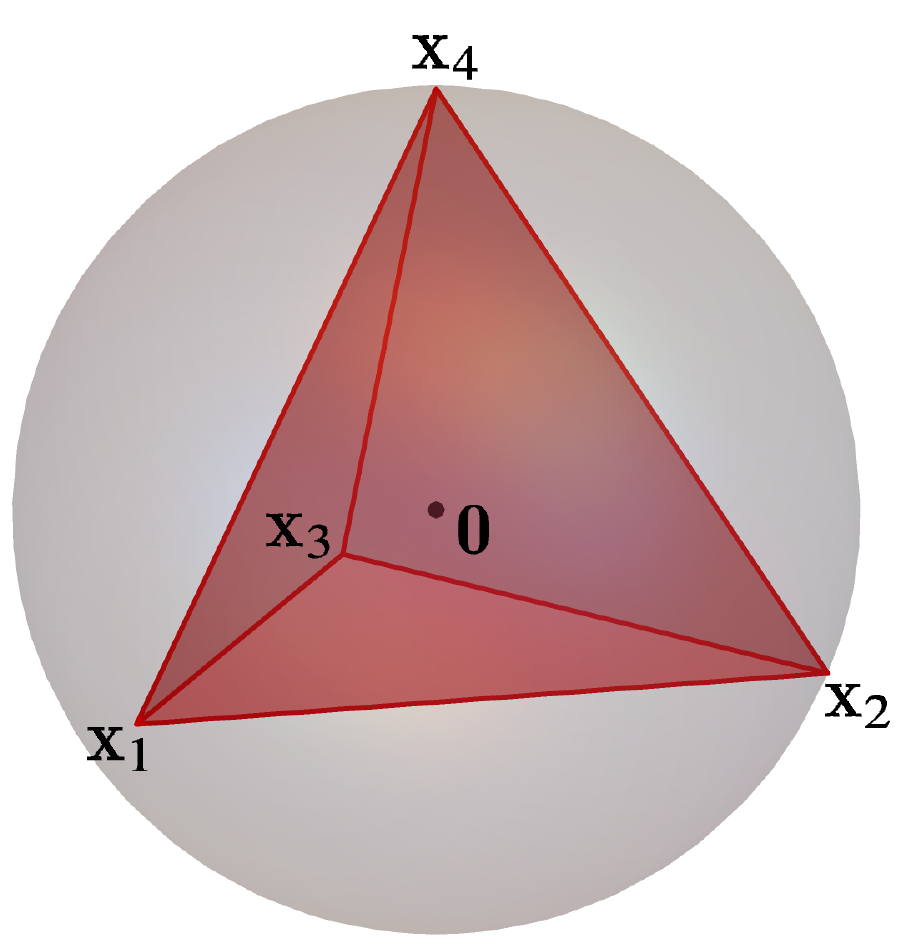}
  \caption{Tetrahedron relative equilibrium on the sphere.
    One may take, e.g., $\mathbf{x}_{1} = (2\sqrt{2}, 0, 1)$, $\mathbf{x}_{2} = (-\sqrt{2}, \sqrt{6}, -1)$, $\mathbf{x}_{3} = (-\sqrt{2}, -\sqrt{6}, -1)$, $\mathbf{x}_{4} = (0, 0, 3)$ with $R = 3$.}
  \label{fig:tetrahedron}
\end{figure}

Using our shape variables, let us set
\begin{equation*}
  |\lambda_{14}| = |\lambda_{24}| = |\lambda_{34}| = \frac{2}{\sqrt{3}},
  \qquad
  \mu_{12} = -\mu_{13} = \mu_{23} = \frac{8}{3\sqrt{3}}\,\rmi.
\end{equation*}
Notice that $\Im\mu_{13}$ is the negative of $\Im\mu_{12}$ and $\Im\mu_{23}$ because the orientation of the triangle formed by $(\mathbf{x}_{1}, \mathbf{x}_{3}, \mathbf{x}_{4})$ is the opposite of those by $(\mathbf{x}_{1}, \mathbf{x}_{2}, \mathbf{x}_{4})$ and $(\mathbf{x}_{2}, \mathbf{x}_{3}, \mathbf{x}_{4})$ as one can see (from the origin) in \Cref{fig:tetrahedron}.
It is easy to check that we then have
\begin{equation*}
  \ell_{12} = \ell_{13} = \ell_{14} = \ell_{23} = \ell_{24} = \ell_{34} = 2\sqrt{\frac{2}{3}}R.
\end{equation*}

\subsection{Stability of Tetrahedron Relative Equilibria}
We would like to find a sufficient condition for stability of the tetrahedron relative equilibria.
To our knowledge, existing stability results for tetrahedron equilibria are limited to some special cases:
(i)~Identical vortices, i.e., $\Gamma_{1} = \Gamma_{2} = \Gamma_{3} = \Gamma_{4}$; see \citet{Ku2004} and \citet{MeNeOs2010}.
(ii)~Linear stability condition with $\Gamma_{1} = \Gamma_{2} = -\Gamma_{3} = -\Gamma_{4}$; see \citet[Theorem~4.6]{La2001}.
(iii)~Lyapunov stability for $\Gamma_{1} = \kappa \neq 0$ and $\Gamma_{2} = \Gamma_{3} = \Gamma_{4} = 1$ with $\kappa > 0$ and linear instability with $\kappa < 0$; see \citet[Discussion of the case $n=3$ on p.~465]{LaMoRo2011}.

We would like to generalize some of these results to the non-identical case with $N = 4$ as follows:

\begin{proposition}
  \label{prop:stability}
  The tetrahedron configuration of four point vortices on the sphere is a stable equilibrium of the shape dynamics if all the circulations $\{ \Gamma_{i} \}_{i=1}^{4}$ have the same sign.
\end{proposition}
\begin{proof}
  First recall from \Cref{thm:main}~(\ref{thm:main-f}) that the variables $\{ \Re\mu_{ij} \}_{1\le i < j \le 3}$ depend on the rest of the variables implicitly as 
  \begin{equation*}
    f_{ij}(\zeta) \defeq \Re\mu_{ij} - \frac{ (\Re\mu_{ij})^{2} + (\Im\mu_{ij})^{2} }{ |\lambda_{iN}|^{2} |\lambda_{jN}|^{2} } - |\lambda_{iN}|^{2} - |\lambda_{jN}|^{2} + 4
    = 0
    \quad
    \text{with}
    \quad
    1\le i < j \le 3.
  \end{equation*}

  The Hamiltonian for the shape dynamics is
  \begin{equation*}
    \mathcal{H}(\zeta) \defeq -\frac{1}{4\pi R^{2}} \parentheses{
      \sum_{1\le i < j \le 3} \Gamma_{i} \Gamma_{j} \ln \parentheses{
        R^{2}\parentheses{
          4 - \frac{ |\mu_{ij}|^{2} }{ |\lambda_{i4}|^{2} |\lambda_{j4}|^{2} }
        }
      }
      + \Gamma_{4} \sum_{1\le i \le 3} \Gamma_{i} \ln \parentheses{
        R^{2} (4 - |\lambda_{i4}|^{2})
      }
    },
  \end{equation*}
  and, recall from \Cref{thm:main}~(\ref{thm:main-Casimirs}) (see also \Cref{prop:Casimirs}) that the shape dynamics possesses a family of Casimirs.
  Particularly, we have
  \begin{align*}
    C_{2}(\zeta)
    \defeq \sum_{i=1}^{4}\Gamma_{i}^{2}
    + \frac{1}{2} \sum_{1 \le i < j \le 3} \Gamma_{i} \Gamma_{j} \frac{ |\mu_{ij}|^{2} }{ |\lambda_{i4}|^{2} |\lambda_{j4}|^{2} }
    + \frac{\Gamma_{4}}{2} \sum_{1 \le i \le 3} \Gamma_{i} \abs{\lambda_{i4}}^{2}.
  \end{align*}

  We would like to use the energy--Casimir method using $\{ f_{ij} \}_{1\le i < j \le 3}$ as additional invariants as well.
  Specifically, let us write the tetrahedron (relative) equilibrium as 
  \begin{equation*}
    \zeta_{\rm e} \defeq \parentheses{
      \frac{2}{\sqrt{3}},\,
      \frac{2}{\sqrt{3}},\,
      \frac{2}{\sqrt{3}},\,
      \frac{8}{3\sqrt{3}}\,\rmi,\,
      -\frac{8}{3\sqrt{3}}\,\rmi,\,
      \frac{8}{3\sqrt{3}}\,\rmi
    },
  \end{equation*}
  and seek a Lyapunov function of the form
  \begin{equation*}
    \mathcal{E}(\zeta)
    \defeq \mathcal{H}(\zeta)
    + \frac{1}{\pi R^{2}} \parentheses{
      \frac{1}{8}\Phi(C_{2}(\zeta) - C_{2}|_{\rm e})
      + \frac{3}{256}\Psi(f_{12}(\zeta), f_{13}(\zeta), f_{23}(\zeta))
    }
  \end{equation*}
  with some smooth functions $\Phi\colon \R \to \R$ and $\Psi\colon \R^{3} \to \R$, where $C_{2}|_{\rm e} \defeq C_{2}(\zeta_{\rm e})$.
  Then $\mathcal{E}$ is an invariant of the shape dynamics because $\mathcal{H}$, $C_{2}$, and $\{ f_{ij} \}_{1\le i < j \le 3}$ are all invariants.
  
  It is then a straightforward computation to show that the gradient $D\mathcal{E}(\zeta_{\rm e})$ vanishes if
  \begin{equation}
    \label{eq:first_variation}
    \Phi'(0) = -\frac{3}{2},
    \qquad
    D\Psi(0) = 0.
  \end{equation}
  In order to show that $\zeta_{\rm e}$ is stable, it suffices to show that there exist $\Phi$ and $\Psi$ such that \eqref{eq:first_variation} holds and also the Hessian $D^{2}\mathcal{E}(\zeta_{\rm e})$ is positive definite.  
  To that end, set
  \begin{equation*}
    A \defeq \frac{256}{9} \pi R^{2}\, D^{2}\mathcal{E}(\zeta_{\rm e}),
  \end{equation*}
  and let $D^{2}_{ij}\Psi$ denote the second derivative $D_{i}D_{j}\Psi$ with $i, j \in \{1, 2, 3\}$.
  Assuming that $\Phi''(0) = 0$ and $D^{2}_{ij}\Psi(0) = 0$ for $i \neq j$, the leading principal minors $\{ d_{i} \}_{i=1}^{9}$ of $A$ satisfy
  \begin{gather*}
    d_{1} = \Gamma_{1} (\Gamma_{2} + \Gamma_{3} + \Gamma_{4}),
    \qquad
    d_{2} = \Gamma_{1} \Gamma_{2} (\Gamma_{3} + \Gamma_{4}) (\Gamma_{1} + \Gamma_{2} + \Gamma_{3} + \Gamma_{4}),
    \\
    d_{3} = \Gamma_{1} \Gamma_{2} \Gamma_{3} \parentheses{
      \Gamma_{1}^{2} \Gamma_{4}
      + \Gamma_{4} (\Gamma_{2} + \Gamma_{3} + \Gamma_{4})^{2}
      + 2\Gamma_{1} ( \Gamma_{4} (\Gamma_{3} + \Gamma_{4}) + \Gamma_{2}  (2\Gamma_{3} + \Gamma_{4}) )
    },
  \end{gather*}
  and
  \begin{gather*}
    \frac{d_{4}}{d_{3}} = \frac{1}{3} D^{2}_{11}\Psi(0),
    \qquad
    \frac{d_{5}}{d_{2}} = \Gamma_{1} \Gamma_{2} \Gamma_{3} \Gamma_{4}\, D^{2}_{11}\Psi(0),
    \qquad
    \frac{d_{6}}{d_{5}} = \frac{1}{3} D^{2}_{22}\Psi(0),
    \\
    \frac{d_{7}}{d_{1}} = \parentheses{ \Gamma_{1} \Gamma_{2} \Gamma_{3} \Gamma_{4} }^{2} D^{2}_{11}\Psi(0)\, D^{2}_{22}\Psi(0),
    \qquad
    \frac{d_{8}}{d_{7}} = \frac{1}{3} D^{2}_{33}\Psi(0),
    \\
    \frac{d_{9}}{d_{7}/d_{1}} = \Gamma_{1} \Gamma_{2} \Gamma_{3} \Gamma_{4}\, D^{2}_{33}\Psi(0).
  \end{gather*}

  Therefore, if $\Gamma_{i} > 0$ (or $\Gamma_{i} < 0$) for every $i \in \{1, 2, 3, 4\}$, then $d_{k} > 0$ for $k \in \{1, \dots, 9\}$, and hence $A$ is positive definite, provided that $D^{2}_{ii}\Psi(0) > 0$ for every $i \in \{1, 2, 3\}$; for example, one may take
  \begin{equation*}
    \Phi(x) = -\frac{3}{2}x,
    \qquad
    \Psi(y_{1}, y_{2}, y_{3}) = y_{1}^{2} + y_{2}^{2} + y_{3}^{2}
  \end{equation*}
  to satisfy the conditions we impose.
\end{proof}

\section{Conclusion and Outlook}
\subsection{Conclusion}
We found the Hamiltonian formulation for the shape dynamics of $N$ point vortices on the sphere by first lifting the dynamics from $\mathbb{S}^{2}_{R}$ to $\C^{2}$ and then applying $\U(2)$-reduction followed by $\mathbb{T}^{N-1}$-reduction, as opposed to performing the (direct) $\SO(3)$-reduction.

The $\U(2)$-reduction was facilitated by a dual pair found by \citet{SkVi2019} and yields a Lie--Poisson dynamics on the dual of the Lie algebra $\u(N)_{\Gamma}$ defined in \Cref{ssec:U(D_Gamma)-u(N)}, whereas the $\mathbb{T}^{N-1}$-reduction was a standard Poisson reduction.
The resulting shape variables give a parametrization of an ambient space of the ($\SO(3)$-reduced) shape space that is hard to parametrize directly.
As a result, our approach yields a concrete expression for the shape dynamics that is difficult to obtain by the $\SO(3)$-reduction.

We also found a family of Casimirs for the shape dynamics exploiting the Lie--Poisson structure.

As an application, we applied our formulation to the tetrahedron relative equilibrium of non-identical vortices, and proved that it is a stable equilibrium of the shape dynamics if all the circulations have the same sign, generalizing some of existing works on the problem.

\subsection{Outlook}
\label{ssec:outlook}
There are several topics to explore for future work that are suggested by the reviewers.

It is interesting to see if our shape variables facilitate explicit integration for those integrable cases, particularly $N = 3$.

In order to better understand the geometry of the shape dynamics, we need to know the geometry of the coadjoint orbit $\mathcal{O} \subset \mathfrak{u}(N)_{\Gamma}^{*}$.
This is fairly simple if all the circulations have the same sign---in which case $\U(D_{\Gamma})$ is isomorphic to $\U(N)$ (and compact), and so $\mathcal{O}$ is essentially the coadjoint orbit in $\mathfrak{u}(N)^{*}$.
However, the general case requires a classification of the coadjoint orbits of the (non-compact) indefinite unitary group $\U(p,q)$ with signature $(p,q)$.
Also, given the compactness of the shape space, this suggests that the momentum map $\mathbf{N}$ is proper, giving an insight into the structure of $\mathcal{O}$.

The linear instability of the tetrahedron relative equilibrium for the special case $\Gamma_{1} = \kappa < 0$ and $\Gamma_{2} = \Gamma_{3} = \Gamma_{4} = 1$ from \cite{LaMoRo2011} also suggests the possibility that the sufficient condition for stability from \Cref{prop:stability} may be also necessary.

It is also an interesting future work to extend our result on the tetrahedron to other Platonic solids with non-identical vortices in order to generalize the results from \cite{Ku2004} and \cite{MeNeOs2010}.

Our lifting of the dynamics to $\C^{2}$ and the $\mathbb{S}^{1}$-bundle structure is reminiscent of the ``post-classical'' formalism of \citet{Tu2003}.
It is interesting to see if there is any relationship between such a lifting and geometric (pre-)quantization.

\appendix

\numberwithin{equation}{section}
\section{Identification of $\su(2)$ with $\R^{3}$}
\label{sec:su2}
\subsection{Isomorphism between $\su(2)$ and $\R^{3}$}
We define a basis $\{ \tau_{i} \}_{i=1}^{3}$ for $\su(2)$ by setting
\begin{equation*}
  \tau_{1} \defeq
   -\frac{\rmi}{2}
  \begin{bmatrix}
    0 & 1 \\
    1 & 0
  \end{bmatrix},
  \qquad
  \tau_{2} \defeq
  -\frac{\rmi}{2}
  \begin{bmatrix}
    0 & -\rmi \\
    \rmi & 0
  \end{bmatrix},
  \qquad
  \tau_{3} \defeq
  -\frac{\rmi}{2}
  \begin{bmatrix}
    1 & 0 \\
    0 & -1
  \end{bmatrix}
\end{equation*}
so that $[\tau_{i}, \tau_{j}] = \tau_{k}$ for every even permutation $(i,j,k)$ of $(1,2,3)$.
We may then identify $\su(2)$ with $\R^{3}$ via the map
\begin{equation}
  \label{eq:R3-su2}
  f\colon \R^{3} \to \su(2);
  \qquad
  \boldsymbol{\xi} = (\xi_{1}, \xi_{2}, \xi_{3})
  \mapsto
  \sum_{j=1}^{3} \xi_{j} \tau_{j} = -\frac{\rmi}{2}
  \begin{bmatrix}
    \xi_{3} & \xi_{1} - \rmi\xi_{2} \\
    \xi_{1} + \rmi\xi_{2} & -\xi_{3}
  \end{bmatrix}.
\end{equation}
The inner product on $\su(2)$ is inherited from $\u(2)$:
\begin{equation}
  \label{eq:ip-su2}
  \ip{\xi}{\eta} \defeq 2\tr(\xi^{*} \eta) = \boldsymbol{\xi} \cdot \boldsymbol{\eta},
\end{equation}
i.e., it is compatible with the standard dot product in $\R^{3}$ under the above identification.

It is also straightforward to see that the commutator in $\su(2)$ is compatible with the cross product in $\R^{3}$ in the sense that, for all $\boldsymbol{\xi}, \boldsymbol{\eta} \in \R^{3}$,
\begin{equation}
  \label{eq:f-bracket}
  f(\boldsymbol{\xi} \times \boldsymbol{\eta}) = [f(\boldsymbol{\xi}), f(\boldsymbol{\eta})]
  .
\end{equation}
That is, $f$ gives a Lie algebra isomorphism between $\R^{3}$ and $\su(2)$.

\subsection{The $\ad$ and $\ad^{*}$ operators in $\su(2)$ and $\R^{3}$}
The property~\eqref{eq:f-bracket} indicates that
\begin{equation*}
  f(\boldsymbol{\xi} \times \boldsymbol{\eta}) = [\xi, \eta] = \ad_{\xi}\eta.
\end{equation*}
Hence for every $\mu \in \su(2)^{*}$, we have
\begin{equation*}
  \ip{ \ad_{\xi}^{*}\mu }{ \eta }
  = \ip{ \mu }{ \ad_{\xi}\eta }
  = \boldsymbol{\mu} \cdot ( \boldsymbol{\xi} \times \boldsymbol{\eta} )
  = ( \boldsymbol{\mu} \times \boldsymbol{\xi} ) \cdot \boldsymbol{\eta}.
\end{equation*}
So we have
\begin{equation*}
  \ad_{\xi}^{*}\mu = \boldsymbol{\mu} \times \boldsymbol{\xi}
\end{equation*}
under the above identification of $\su(2)$ and $\su(2)^{*}$ with $\R^{3}$.

\section{Recovering Original Dynamics from Lifted Dynamics}
\label{sec:recovering}
This appendix gives a proof of \Cref{prop:reduction_of_lifted_dynamics}.
The main idea is to use a dual pair to find a natural parametrization of the Hopf fibration $\mathbb{S}^{3}_{\sqrt{R}} \to \mathbb{S}^{2}_{R}$.

\subsection{Symplectic Reduction via Dual Pair}
Since the reduction is performed for each copy of $\C^{2}$ in $(\C^{2})^{N}$ separately, we first perform the reduction for a single copy of $\C^{2}$.
Hence the momentum map $\mathbf{J}$ is, dropping the subscripts for $\varphi$,
\begin{equation*}
  \mathbf{J}\colon \C^{2} \to \R;
  \qquad
  \mathbf{J}(\varphi) = -\frac{2}{R}\Gamma \norm{\varphi}^{2}.
\end{equation*}

In order to construct a dual pair, we also define an action of $\SU(2)$ on $\C^{2}$ as follows:
\begin{equation*}
  \SU(2) \times \C^{2} \to \C^{2};
  \qquad
  (U,\varphi) \mapsto U\varphi.
\end{equation*}
For every $\xi \in \su(2)$, its corresponding infinitesimal generator is 
\begin{equation*}
  \xi_{\C^{2}}(\varphi) = \xi \varphi.
\end{equation*}
We identify $\su(2)^{*}$ with $\su(2)$ via the inner product~\eqref{eq:ip-su2} from \Cref{sec:su2}.
Then the associated momentum map $\mathbf{M}\colon \C^{2} \to \su(2)^{*}$ satisfies
\begin{align*}
  \ip{ \mathbf{M}(\varphi) }{ \xi }
  &= \Theta(\varphi) \cdot \xi_{\C^{2}}(\varphi) \\
  &= -\frac{2}{R}\Gamma \Im\parentheses{ \varphi^{*} \xi \varphi } \\
  &= \frac{2}{R} \rmi\, \Gamma \varphi^{*} \xi \varphi \\
  &= \frac{2}{R} \rmi\, \Gamma \tr( \varphi \varphi^{*} \xi ) \\
  &= \frac{2}{R} \tr\parentheses{ (-\rmi \Gamma \varphi \varphi^{*})^{*} \xi } \\
  &= \ip{ -\frac{\rmi}{R} \Gamma \varphi \varphi^{*} }{ \xi }.
\end{align*}
However, since $\mathbf{M}$ takes values in $\su(2)^{*} \cong \su(2)$, we have
\begin{align}
  \mathbf{M}(\varphi)
  &= -\frac{\rmi}{R} \Gamma \parentheses{ \varphi \varphi^{*} - \frac{1}{2}\tr(\varphi \varphi^{*}) I } \nonumber\\
  &= -\frac{\rmi}{R} \Gamma \parentheses{ \varphi \varphi^{*} - \frac{1}{2}\norm{\varphi}^{2} I } \nonumber\\
  &= -\frac{\rmi}{2R} \Gamma
    \begin{bmatrix}
      |z|^{2} - |u|^{2} & 2z\bar{u} \\
      2\bar{z}u & |u|^{2} - |z|^{2}
    \end{bmatrix} \nonumber\\
  &= \frac{\Gamma}{R}\parentheses{
    2\Re(\bar{z}u),\, 2\Im(\bar{z}u),\, |z|^{2} - |u|^{2}
    },
  \label{eq:M}
\end{align}
where we used the identification~\eqref{eq:R3-su2} of $\su(2)$ with $\R^{3}$ from \Cref{sec:su2}.
Notice that the above expression for $\mathbf{M}$ essentially gives the Hopf fibration.

As a result, we have a pair of momentum maps defined on $\C^{2}$:
\begin{equation*}
  \begin{tikzcd}
    \R & \C^{2} \arrow[swap]{l}{\mathbf{J}} \arrow{r}{\mathbf{M}} & \su(2)^{*}.
  \end{tikzcd}
\end{equation*}
This pair of momentum maps is known to constitute a dual pair; see \citet{GoStMa1987} and \citet{HoVi2012}.
This implies that the Marsden--Weinstein quotient $\mathbf{J}^{-1}(-2\Gamma)/\mathbb{S}^{1}$ is symplectomorphic to a coadjoint orbit $\mathcal{O}$ in $\su(2)^{*}$.
More specifically, the momentum map $\mathbf{M}$ restricted to the level set $\mathbf{J}^{-1}(-2\Gamma)$ gives rise to the symplectomorphism $\overline{\mathbf{M}}\colon \mathbf{J}^{-1}(-2\Gamma)/\mathbb{S}^{1} \to \mathcal{O}$, where $\mathcal{O}$ is equipped with the $(+)$-Kirillov--Kostant--Souriau (KKS) symplectic structure (see, e.g., \citet[Chapter~1]{Ki2004} and \citet[Chapter~14]{MaRa1999} and references therein): For every $\mu \in \mathcal{O}$ and $\xi, \eta \in \su(2)$,
\begin{equation*}
  \Omega_{\mathcal{O}}(\mu)\parentheses{ \ad_{\xi}^{*}\mu, \ad_{\eta}^{*}\mu } \defeq \ip{\mu}{[\xi,\eta]},
\end{equation*}
or using the identification between $\su(2)^{*} \cong \su(2)$ with $\R^{3}$ in \eqref{eq:R3-su2},
\begin{equation*}
  \Omega_{\mathcal{O}}(\boldsymbol{\mu})\parentheses{ \boldsymbol{\mu} \times \boldsymbol{\xi}, \boldsymbol{\mu} \times \boldsymbol{\eta} } = \boldsymbol{\mu} \cdot (\boldsymbol{\xi} \times \boldsymbol{\eta})
  \quad
  \text{or}
  \quad
  \Omega_{\mathcal{O}}(\boldsymbol{\mu})( \boldsymbol{v}, \boldsymbol{w} ) = \frac{\boldsymbol{\mu}}{ |\boldsymbol{\mu}|^{2} } \cdot (\boldsymbol{v} \times \boldsymbol{w})
\end{equation*}
In other words, the collectivization by $\mathbf{M}$ coincides with the symplectic reduction by $\mathbb{S}^{1}$.

It is well known that $\mathcal{O}$ is a two-dimensional sphere as well.
One can also see it from the expression~\eqref{eq:M} that if $\varphi \in \mathbf{J}^{-1}(-2\Gamma)$, then $\norm{\varphi} = \sqrt{R}$, and so $\mathbf{M}(\varphi)$, as a vector in $\R^{3}$, is in the sphere with radius $\Gamma$ centered at the origin.
Hence $\mathbf{J}^{-1}(-2\Gamma)/\mathbb{S}^{1}$ is a sphere as well.

In order to show that the reduced dynamics is indeed the point vortex dynamics on $\mathbb{S}^{2}_{R}$, we identify $\mathbf{J}^{-1}(-2\Gamma)/\mathbb{S}^{1}$ with $\mathbb{S}^{2}_{R}$ via
\begin{equation*}
  \overline{\mathbf{M}}\colon \mathbf{J}^{-1}(-2\Gamma)/\mathbb{S}^{1} \cong \mathbb{S}^{2}_{R} \to \mathcal{O} \cong \mathbb{S}^{2}_{\Gamma};
  \qquad
  \mathbf{x} \mapsto \frac{\Gamma}{R} \mathbf{x}.
\end{equation*}
In other words, we are setting $\frac{\Gamma}{R} \mathbf{x} = \mathbf{M}(\varphi)$.
Then, in view of \eqref{eq:M}, we have
\begin{equation*}
  \mathbf{x} = \parentheses{ 2\Re(\bar{z}u),\, 2\Im(\bar{z}u),\, |z|^{2} - |u|^{2} },
\end{equation*}
that is, $\varphi$ and $\mathbf{x}$ are related via the Hopf fibration as in \cite{VaLe2014}.

Pulling back $\Omega_{\mathcal{O}}$ to $\mathbb{S}^{2}_{R}$ by $\overline{\mathbf{M}}$, we obtain
\begin{equation}
  \label{eq:Omega0-S^2}
  \overline{\mathbf{M}}^{*}\Omega_{\mathcal{O}}(\mathbf{x})(\mathbf{v}, \mathbf{w}) = \frac{\Gamma}{R} \mathbf{x} \cdot (\mathbf{v} \times \mathbf{w}),
\end{equation}
which is the area form on $\mathbb{S}^{2}_{R}$ multiplied by $\Gamma$.
The corresponding Poisson bracket is, for all smooth $F, H\colon \R^{3} \to \R$,
\begin{equation*}
  \PB{F}{H}_{\R^{3}}(\mathbf{x}) = \frac{R}{\Gamma} \mathbf{x} \cdot \parentheses{ \pd{F}{\mathbf{x}} \times \pd{H}{\mathbf{x}} }.
\end{equation*}

\subsection{Vortex Dynamics on $\mathbb{S}^{2}_{R}$ and Lie--Poisson Equation}
Now let us come back to the lifted dynamics of $N$ vortices in $(\C^{2})^{N}$.
The above argument applies to each copy of $\C^{2}$, and so we have the momentum map
\begin{equation*}
  \mathbf{M}\colon (\C^{2})^{N} \to (\su(2)^{*})^{N} \cong (\R^{3})^{N};
  \qquad
  (\varphi_{1}, \dots, \varphi_{N}) \mapsto \frac{1}{R}(\Gamma_{1}\mathbf{x}_{1}, \dots, \Gamma_{N}\mathbf{x}_{N}),
\end{equation*}
where, for each $i \in \{1, \dots, N\}$,
\begin{equation}
  \label{eq:C^2-R^3}
  \mathbf{x}_{i} \defeq \parentheses{ 2\Re(\bar{z}_{i}u_{i}),\, 2\Im(\bar{z}_{i}u_{i}),\, |z_{i}|^{2} - |u_{i}|^{2} } \in \R^{3}.
\end{equation}
Since each copy of $\mathbb{S}^{2}_{R}$ is equipped with the symplectic form~\eqref{eq:Omega0-S^2}, this gives rise to the symplectic form~\eqref{eq:Omega-S^2} on $(\mathbb{S}^{2}_{R})^{N}$.
The corresponding  Poisson bracket on $(\R^{3})^{N}$ is then \eqref{eq:PB-R3}.

One then sees that $H_{\R^{3}}\colon (\R^{3})^{N} \to \R$ from \eqref{eq:H-R^3} is the collective Hamiltonian, i.e., $\bar{H} \circ \mathbf{M} = H$ where $H$ is defined in \eqref{eq:H-C^2}.
As a result, the reduced dynamics is the Hamiltonian system~\eqref{eq:vortices_on_sphere-Hamiltonian}.
This completes the proof of \Cref{prop:reduction_of_lifted_dynamics}.

\section{Proof of \Cref{prop:Casimirs}}
\label{sec:Proof-Casimirs}
\begin{enumerate}[(i)]
\item Let $j \in \N$ be arbitrary.
  Let us first show that $C_{j}$ is indeed a real-valued function:
  For every $\lambda \in \u(N)_{\Gamma}^{*}$, we have
  \begin{align*}
    \overline{C_{j}(\lambda)}
    &= \tr\parentheses{ (-\rmi D_{\Gamma} \overline{\lambda})^{j} } \\
    &= \tr\brackets{ \parentheses{ (-\rmi D_{\Gamma} \overline{\lambda})^{T} }^{j} } \\
    &= \tr\parentheses{ (-\rmi \lambda^{*} D_{\Gamma} )^{j} } \\
    &= \tr\parentheses{ (\rmi \lambda D_{\Gamma} )^{j} } \\
    &= \tr\parentheses{ (\rmi D_{\Gamma} \lambda)^{j} } = C_{j}(\lambda).
  \end{align*}
  We also see that $C_{j}$ is $\Ad^{*}$-invariant:
  For every $U \in \U(D_{\Gamma})$, we have
  \begin{align*}
    C_{j}(\Ad_{U}^{*}\lambda)
    &= \tr\parentheses{ (\rmi D_{\Gamma} U^{*} \lambda U)^{j} } \\
    &= \tr\parentheses{ (\rmi U D_{\Gamma} U^{*} \lambda)^{j} } \\
    &= \tr\parentheses{ (\rmi D_{\Gamma} \lambda)^{j} } = C_{j}(\lambda).
  \end{align*}
  Since any $\Ad^{*}$-invariant differentiable function is a Casimir (see, e.g., \citet[Corollary~14.4.3]{MaRa1999}), this implies that $C_{j}$ is a Casimir function for the Lie--Poisson bracket~\eqref{eq:LPB-u(N)^*}.
\item Let us set $A \defeq \rmi D_{\Gamma} \lambda$ so that we have $C_{j}(\lambda) = \tr\parentheses{ A^{j} }$.
  By the Cayley--Hamilton Theorem, we have $p(A) = 0$, where $p$ is the characteristic polynomial of $A$:
  \begin{equation*}
    p(x) \defeq \det\parentheses{ x I - A }
    = x^{N} - c_{1} x^{N-1} - c_{2} x^{N-2} - \dots - c_{N},
  \end{equation*}
  where the coefficients $\{ c_{k}(\lambda) \}_{k=1}^{N}$ are determined by the Faddeev--LeVerrier algorithm (see, e.g., \citet[p.~87]{Ga2000}):
  \begin{align*}
    A_{1} &\defeq A,       & c_{1} &= \tr(A_{1}),              & B_{1} &\defeq A_{1} - c_{1}I, \\
    A_{2} &\defeq A B_{1},  & c_{2} &= \frac{1}{2}\tr(A_{2}),  & B_{2} &\defeq A_{2} - c_{2}I, \\
          && \vdots &&& \\
    A_{N} &\defeq A B_{N-1},  & c_{N} &= \frac{1}{N}\tr(A_{N}). & &
  \end{align*}
  Specifically, this implies that each $c_{j}$ with $j \in \{1, \dots, N\}$ depends on $A$ as a smooth function of $\braces{ \tr\parentheses{ A^{k} } }_{k=1}^{j}$.
  We also obtain the expression
  \begin{equation*}
    c_{N} = \frac{1}{N}\tr(A^{N}) + \dots,
  \end{equation*}
  where the remaining terms do not contain $\tr(A^{N})$.
  Now, taking the trace of
  \begin{equation*}
    p(A) = A^{N} - c_{1} A^{N-1} - c_{2} A^{N-2} - \dots - A_{N} = 0,
  \end{equation*}
  we have
  \begin{equation*}
    \tr\parentheses{ A^{N} } - c_{1} \tr\parentheses{ A^{N-1} } - c_{2} \tr\parentheses{ A^{N-2} } - \dots - c_{N} = 0.
  \end{equation*}
  Since each $c_{j}$ with $j \in \{1, \dots, N-1\}$ depends on $A$ as a smooth function of $\braces{ \tr\parentheses{ A^{k} } }_{k=1}^{j}$, and $c_{N}$ takes the form shown above, $\tr(A^{N})$ can be expressed in terms of $\braces{ \tr\parentheses{ A^{j} } }_{j=1}^{N-1}$.
  It implies that $C_{N}(\lambda)$ can be expressed in terms of $\{ C_{j}(\lambda) \}_{j=1}^{N-1}$.
  This argument extends to $C_{j}(\lambda)$ for $j \ge N+1$ recursively.
\end{enumerate}

\section{Vector Identities in $\C^{2}$ and $\R^{3}$}
Since we use the lifted vortex dynamics in $\C^{2}$ to describe the dynamics in $\mathbb{S}^{2}_{R} \subset \mathbb{R}^{3}$, we make use of some identities that hold between vectors in $\C^{2}$ and those in $\R^{3}$ via the map~\eqref{eq:C^2-R^3}.
This appendix presents detailed derivations of these identities, because the derivations are, although straightforward, quite cumbersome and non-trivial, and also because there does not seem to be proper references on these identities.

\subsection{Vectors in $\C^{2}$ and $\R^{3}$}
Recall that, for $i \in \{1, \dots, N\}$,  we let $\varphi_{i} = (z_{i}, u_{i}) \in \C^{2}$ and set
\begin{equation}
  \tag{\ref{eq:C^2-R^3}}
  \mathbf{x}_{i} \defeq \parentheses{ 2\Re(\bar{z}_{i}u_{i}),\, 2\Im(\bar{z}_{i}u_{i}),\, |z_{i}|^{2} - |u_{i}|^{2} } \in \R^{3}.
\end{equation}
We would like to derive those formulas for vectors in $\C^{2}$ that give some familiar objects in vector algebra in $\R^{3}$.

\subsection{Inner Product in $\C^{2}$ and Dot Product in $\R^{3}$}
\label{ssec:ip-C^2_R^3}
The dot product in $\R^{3}$ is related to the inner product in $\C^{2}$ as follows:
\begin{align}
  \mathbf{x}_{1} \cdot \mathbf{x}_{2}
  &= 4\parentheses{
    \Re(\bar{z}_{1}u_{1}) \Re(\bar{z}_{2}u_{2}) + \Im(\bar{z}_{1}u_{1}) \Im(\bar{z}_{2}u_{2})
    }
    + \parentheses{ |z_{1}|^{2} - |u_{1}|^{2} }\parentheses{ |z_{2}|^{2} - |u_{2}|^{2} } \nonumber\\
  &= 4 \Re(\bar{z}_{1}u_{1}z_{2}\bar{u}_{2})
    + \parentheses{ |z_{1}|^{2} - |u_{1}|^{2} }\parentheses{ |z_{2}|^{2} - |u_{2}|^{2} } \nonumber\\
  &= 2 (\bar{z}_{1}u_{1}z_{2}\bar{u}_{2} + z_{1}\bar{u}_{1}\bar{z}_{2}u_{2})
    + \parentheses{ |z_{1}|^{2} - |u_{1}|^{2} }\parentheses{ |z_{2}|^{2} - |u_{2}|^{2} } \nonumber\\
  &= 2 \parentheses{
    \bar{z}_{1}u_{1}z_{2}\bar{u}_{2} + z_{1}\bar{u}_{1}\bar{z}_{2}u_{2}
    + |z_{1}|^{2}|z_{2}|^{2} + |u_{1}|^{2}|u_{2}|^{2}
    }
    - \parentheses{ |z_{1}|^{2} + |u_{1}|^{2} }\parentheses{ |z_{2}|^{2} + |u_{2}|^{2} } \nonumber\\
  &= 2 (\bar{z}_{1}z_{2} + \bar{u}_{1}u_{2})(z_{1}\bar{z}_{2} + u_{1}\bar{u}_{2})
    - \parentheses{ |z_{1}|^{2} + |u_{1}|^{2} }\parentheses{ |z_{2}|^{2} + |u_{2}|^{2} } \nonumber\\
  &= 2 | \bar{z}_{1}z_{2} + \bar{u}_{1}u_{2} |^{2}
    - \parentheses{ |z_{1}|^{2} + |u_{1}|^{2} }\parentheses{ |z_{2}|^{2} + |u_{2}|^{2} } \nonumber\\
  &= 2 | \varphi_{1}^{*}\varphi_{2} |^{2}
    - \norm{ \varphi_{1} }^{2} \norm{ \varphi_{2} }^{2}.
    \label{eq:dot_products}
\end{align}
Hence we have
\begin{equation*}
  |\mathbf{x}_{1}|^{2} = \norm{ \varphi_{1} }^{4}.
\end{equation*}
This implies that that the three sphere with radius $\sqrt{R}$ is mapped to the two-sphere with radius $R$ (both centered at the origin) under the map~\eqref{eq:C^2-R^3}.

We also have 
\begin{align}
  | \mathbf{x}_{1} - \mathbf{x}_{2} |^{2}
  &= |\mathbf{x}_{1}|^{2} + |\mathbf{x}_{2}|^{2} - 2\mathbf{x}_{1} \cdot \mathbf{x}_{2} \nonumber\\
  &= \norm{ \varphi_{1} }^{4} + \norm{ \varphi_{2} }^{4} - 4 | \varphi_{1}^{*}\varphi_{2} |^{2}
    + 2 \norm{ \varphi_{1} }^{2} \norm{ \varphi_{2} }^{2} \nonumber\\
  &= \left( \norm{\varphi_{1}}^{2} + \norm{\varphi_{2}}^{2} \right)^{2}
    - 4| \varphi_{1}^{*}\varphi_{2} |^{2},
    \label{eq:C^2_distance-R^3}    
\end{align}
and so, if $\varphi_{1}, \varphi_{2} \in \mathbb{S}^{3}_{\sqrt{R}}$, then $\mathbf{x}_{1}, \mathbf{x}_{2} \in \mathbb{S}^{2}_{R}$, and
\begin{equation}
  \label{eq:S^3_distance-S^2}
  | \mathbf{x}_{1} - \mathbf{x}_{2} |^{2}
  = 4\parentheses{ R^{2} - | \varphi_{1}^{*}\varphi_{2} |^{2} }.
\end{equation}

\subsection{Triple Product in $\R^{3}$}
We have
\begin{align*}
  (\varphi_{1}^{*}\varphi_{2}) (\varphi_{3}^{*}\varphi_{1}) (\varphi_{2}^{*}\varphi_{3})
  &= (\bar{z}_{1}z_{2} + \bar{u}_{1}u_{2}) (\bar{z}_{3}z_{1} + \bar{u}_{3}u_{1}) (\bar{z}_{2}z_{3} + \bar{u}_{2}u_{3}) \\
  % &= \parentheses{
  %   |z_{1}|^{2}z_{2}\bar{z}_{3} + \bar{z}_{1}z_{2}\bar{u}_{3}u_{1} + \bar{z}_{3}z_{1}\bar{u}_{1}u_{2} + |u_{1}|^{2}u_{2}\bar{u}_{3}
  %   } (\bar{z}_{2}z_{3} + \bar{u}_{2}u_{3}) \\
  &= |z_{1}|^{2}|z_{2}|^{2}|z_{3}|^{2} + |u_{1}|^{2}|u_{2}|^{2}|u_{3}|^{2} \\
  &\quad + \sum_{(i,j,k)\in \mathcal{Z}_{3}} \parentheses{
    |z_{i}|^{2}z_{j}\bar{z}_{k}\bar{u}_{j}u_{k} + |u_{i}|^{2}\bar{z}_{j}z_{k}u_{j}\bar{u}_{k}
    },
  % &\quad + |z_{1}|^{2}z_{2}\bar{z}_{3}\bar{u}_{2}u_{3} + |u_{1}|^{2}\bar{z}_{2}z_{3}u_{2}\bar{u}_{3} \\
  % &\quad + |z_{2}|^{2}z_{3}\bar{z}_{1}\bar{u}_{3}u_{1} + |u_{2}|^{2}\bar{z}_{3}z_{1}u_{3}\bar{u}_{1} \\
  % &\quad + |z_{3}|^{2}z_{1}\bar{z}_{2}\bar{u}_{1}u_{2} + |u_{3}|^{2}\bar{z}_{1}z_{2}u_{1}\bar{u}_{2},
\end{align*}
where $\mathcal{Z}_{3}$ is the set of all cyclic permutations of $(1,2,3)$, i.e., $\mathcal{Z}_{3} \defeq \{(1,2,3), (2,3,1), (3,1,2)\}$.

However,
\begin{align*}
  |z_{1}|^{2}&|z_{2}|^{2}|z_{3}|^{2} + |u_{1}|^{2}|u_{2}|^{2}|u_{3}|^{2} \\
  &= \prod_{i=1}^{3}\parentheses{ \frac{1}{2}(|z_{i}|^{2} + |u_{i}|^{2}) + \frac{1}{2}(|z_{i}|^{2} - |u_{i}|^{2}) } 
    + \prod_{i=1}^{3}\parentheses{ \frac{1}{2}(|z_{i}|^{2} + |u_{i}|^{2}) - \frac{1}{2}(|z_{i}|^{2} - |u_{i}|^{2}) } \\
  &= \frac{1}{4} \parentheses{
    \prod_{i=1}^{3} (|z_{i}|^{2} + |u_{i}|^{2})
    + \sum_{(i,j,k)\in \mathcal{Z}_{3}}(|z_{i}|^{2} + |u_{i}|^{2}) (|z_{j}|^{2} - |u_{j}|^{2}) (|z_{k}|^{2} - |u_{k}|^{2})
    } \\
  &= \frac{1}{4}\parentheses{
    |\mathbf{x}_{1}| |\mathbf{x}_{2}| |\mathbf{x}_{3}|
    + \sum_{(i,j,k)\in \mathcal{Z}_{3}} |\mathbf{x}_{i}| x_{j}^{3}x_{k}^{3}
    }.
\end{align*}
We also have
\begin{align*}
  \Re(z_{j}\bar{z}_{k}\bar{u}_{j}u_{k})
  &= \Re\parentheses{ \overline{ \bar{z}_{j}{u}_{j} } \bar{z}_{k}u_{k} } \\
  &= \Re(\bar{z}_{j}{u}_{j}) \Re(\bar{z}_{k}u_{k}) + \Im(\bar{z}_{j}{u}_{j}) \Im(\bar{z}_{k}u_{k}) \\
  &= \frac{1}{4} (x_{j}^{1} x_{k}^{1} + x_{j}^{2} x_{k}^{2}) \\
  &= \frac{1}{4} (\mathbf{x}_{j} \cdot \mathbf{x}_{k} - x_{j}^{3} x_{k}^{3})
\end{align*}
and
\begin{align*}
  \Im(z_{j}\bar{z}_{k}\bar{u}_{j}u_{k})
  &= \Im\parentheses{ \overline{ \bar{z}_{j}{u}_{j} } \bar{z}_{k}u_{k} } \\
  &= \Re(\bar{z}_{j}{u}_{j}) \Im(\bar{z}_{k}u_{k}) - \Im(\bar{z}_{j}{u}_{j}) \Re(\bar{z}_{k}u_{k}) \\
  &= \frac{1}{4} (x_{j}^{1} x_{k}^{2} - x_{j}^{2} x_{k}^{1}) \\
  &= \frac{1}{4} (\mathbf{x}_{j} \times \mathbf{x}_{k})^{3},
\end{align*}
where we wrote the components of $\mathbf{x}_{i}$ as $(x_{i}^{1}, x_{i}^{2}, x_{i}^{3})$, and $(\mathbf{x}_{j} \times \mathbf{x}_{k})^{3}$ signifies the third component of $\mathbf{x}_{j} \times \mathbf{x}_{k}$.

As a result, we obtain
\begin{align*}
  \Re\parentheses{
  (\varphi_{1}^{*}\varphi_{2}) (\varphi_{3}^{*}\varphi_{1}) (\varphi_{2}^{*}\varphi_{3})
  }
  &= \frac{1}{4} \parentheses{ |\mathbf{x}_{1}| |\mathbf{x}_{2}| |\mathbf{x}_{3}|
    + \sum_{(i,j,k)\in \mathcal{Z}_{3}} |\mathbf{x}_{i}| x_{j}^{3}x_{k}^{3}
    } \\
  &\quad + \frac{1}{4} \sum_{(i,j,k)\in \mathcal{Z}_{3}} \parentheses{ |z_{i}|^{2} + |u_{i}|^{2} } (\mathbf{x}_{j} \cdot \mathbf{x}_{k} - x_{j}^{3} x_{k}^{3}) \\
  % &\quad + \frac{1}{4} \parentheses{ |z_{1}|^{2} + |u_{1}|^{2} } (\mathbf{x}_{2} \cdot \mathbf{x}_{3} - x_{2}^{3} x_{3}^{3}) \\
  % &\quad + \frac{1}{4} \parentheses{ |z_{2}|^{2} + |u_{2}|^{2} } (\mathbf{x}_{3} \cdot \mathbf{x}_{1} - x_{3}^{3} x_{1}^{3}) \\
  % &\quad + \frac{1}{4} \parentheses{ |z_{3}|^{2} + |u_{3}|^{2} } (\mathbf{x}_{1} \cdot \mathbf{x}_{2} - x_{1}^{3} x_{2}^{3}) \\
  &= \frac{1}{4} \parentheses{
    |\mathbf{x}_{1}| |\mathbf{x}_{2}| |\mathbf{x}_{3}| + \sum_{(i,j,k)\in \mathcal{Z}_{3}}|\mathbf{x}_{i}| \mathbf{x}_{j} \cdot \mathbf{x}_{k}
    }
\end{align*}
and
\begin{align}
  \label{eq:Im(mu_{123})/8}
  \Im\parentheses{
  (\varphi_{1}^{*}\varphi_{2}) (\varphi_{3}^{*}\varphi_{1}) (\varphi_{2}^{*}\varphi_{3})
  }
  &= \frac{1}{4} \parentheses{ |z_{1}|^{2} - |u_{1}|^{2} } (\mathbf{x}_{2} \times \mathbf{x}_{3})^{3} \nonumber\\
  &\quad + \frac{1}{4} \parentheses{ |z_{2}|^{2} - |u_{2}|^{2} } (\mathbf{x}_{3} \times \mathbf{x}_{1})^{3}
    + \frac{1}{4} \parentheses{ |z_{3}|^{2} - |u_{3}|^{2} } (\mathbf{x}_{1} \times \mathbf{x}_{2})^{3} \nonumber\\
  &= \frac{1}{4} \brackets{
    x_{1}^{3} (\mathbf{x}_{2} \times \mathbf{x}_{3})^{3}
    + x_{2}^{3} (\mathbf{x}_{3} \times \mathbf{x}_{1})^{3}
    + x_{3}^{3} (\mathbf{x}_{1} \times \mathbf{x}_{2})^{3}
    } \nonumber\\
  &= \frac{1}{4} \det [\mathbf{x}_{1}\, \mathbf{x}_{2}\, \mathbf{x}_{3}] \nonumber\\
  &= \frac{1}{4} \mathbf{x}_{1} \cdot (\mathbf{x}_{2} \times \mathbf{x}_{3}).
\end{align}

Particularly, if $\mathbf{x}_{i} \in \mathbb{S}^{2}_{R}$ then $\varphi_{i} \in \mathbb{S}^{3}_{\sqrt{R}}$ for $i \in \{1, 2, 3\}$, and so \eqref{eq:dot_products} gives, for all $i, j \in \{1, 2, 3\}$,
\begin{equation*}
  \mathbf{x}_{i} \cdot \mathbf{x}_{j} = 2| \varphi_{i}^{*}\varphi_{j} |^{2} - R^{2}.
\end{equation*}
So we have
\begin{equation}
  \label{eq:Remu_123}
  \begin{split}
    \Re\parentheses{
       (\varphi_{1}^{*}\varphi_{2}) (\varphi_{3}^{*}\varphi_{1}) (\varphi_{2}^{*}\varphi_{3})
     }
    &= \frac{R}{4} \parentheses{
      R^{2} + \mathbf{x}_{1} \cdot \mathbf{x}_{2} + \mathbf{x}_{3} \cdot \mathbf{x}_{1} + \mathbf{x}_{2} \cdot \mathbf{x}_{3}
    } \\
    &= \frac{R}{2} \parentheses{
      \abs{ \varphi_{1}^{*}\varphi_{2} }^{2}
      + \abs{ \varphi_{3}^{*}\varphi_{1} }^{2}
      + \abs{ \varphi_{2}^{*}\varphi_{3} }^{2} - R^{2}
    }.
  \end{split}
\end{equation}

\section*{Acknowledgments}
I would like to thank Melvin Leok for helpful comments and discussions, and the reviewers for their comments and suggestions including the possible future work discussed in \Cref{ssec:outlook}.
This work was supported by NSF grant DMS-2006736.

\bibliography{Point_Vortices-Sphere}
\bibliographystyle{plainnat}

\end{document}